\documentclass{article}
\usepackage{amsmath,amsfonts,graphicx}

\parskip=\medskipamount
 1   
 1  

\newcommand{\R}{\mathbb{R}}
\newcommand{\Ro}{{\cal R}}
\newcommand{\pd}{\partial}

\def\e{\epsilon}
\newcommand{\lag}{\mathfrak{g}}
\def\fpd#1#2{\frac{\partial #1}{\partial #2}}
\newcommand{\F}{\mathbb{F}}

\newtheorem{theorem}{Theorem}
\newtheorem{proposition}{Proposition}
\newtheorem{definition}{Definition}
\newtheorem{remark}{Remark}
\newenvironment{proof}{{\bf Proof.\ }}{}

\begin{document}
\title{Routh reduction for singular Lagrangians}
\author{Bavo Langerock\thanks{BL is an honorary postdoctoral researcher at the Department of Mathematics of Ghent Univeristy and associate academic staff at the Department of Mathematics of KULeuven.}\\ Sint-Lucas school
of Architecture\\ B-9000
Ghent, Belgium\\[1cm]
Marco Castrill\'on L\'opez\\
ICMAT (CSIC-UAM-UC3M-UCM), Dept. Geometr\'{\i}a y Topolog\'{\i}a\\
Facultad de Matem\'aticas\\
Universidad Complutense de Madrid\\ 28040 Madrid, Spain}
\maketitle
\begin{abstract}
\noindent This paper concerns the Routh reduction procedure for
Lagrangians systems with symmetry. It differs from the existing
results on geometric Routh reduction in the fact that no
regularity conditions on either the Lagrangian $L$ or the momentum
map $J_L$ are required apart from the momentum being a regular
value of $J_L$. The main results of this paper are: the
description of a general Routh reduction procedure that preserves
the Euler-Lagrange nature of the original system and the presentation of a
presymplectic framework for Routh reduced systems. In addition, we provide a detailed description and interpretation of the Euler-Lagrange equations for the reduced system. The proposed procedure includes Lagrangian systems with a non-positively
definite kinetic energy metric.

\noindent MSC2000: Primary 70H33. Secondary 70G65, 70H03, 53D20.
\end{abstract}

{\em Keywords:\ }Constrained Lagrangian systems;
momentum map; reduction; Routhian; singular Lagrangians; symmetry;
symplectic form.

\tableofcontents
\markboth{B. Langerock and M. Castrill\'on L\'opez}{Routh reduction for singular Lagrangians}

\section{Introduction}
In geometric accounts to Routh
reduction~\cite{CMR01,mestcram,jalna,marsdenrouth} it is custom to
first consider the restriction of the dynamics associated to a
Lagrangian system with symmetry to the level set of the momentum
map $J^{-1}_L(\mu)$ and then reduce to the quotient manifold under
the action of the isotropy group $G_\mu$. Because this procedure
is very similar to symplectic or Marsden-Weinstein reduction, one often states
that {\em Routh reduction is the Lagrangian analogue of cotangent
bundle reduction} (see e.g.~\cite{jalna,marsdenrouth}). Another
essential ingredient in the work on Routh reduction is the
existence of a regularity condition on the momentum map $J_L$: it
guarantees that the manifold $J^{-1}_L(\mu)/G_{\mu}$ is
diffeomorphic to $T(Q/G)\times Q/G_\mu$.  Roughly said, this
condition provides the manifold $J^{-1}_L(\mu)/G_{\mu}$ with a
tangent bundle structure, an essential feature to reinterpret the
reduced dynamics as being Euler-Lagrange equations. We say that
Lagrangians satisfying this regularity condition are $G$-regular.
For example, the classical $T-V$ Lagrangians fit in this context.
Unfortunately, there are many Lagrangians where this procedure can
not be carried out. This is the case for singular Lagrangians,
though there are also instances of regular Lagrangians where it is
impossible to perform Routhian reduction in this sense. A simple
one is the Lagrangian
$L(q^1,q^2,\dot{q}^1,\dot{q}^2)=(\dot{q}^1)^2+\dot{q}^1\dot{q}^2-V(q^1)$
defined in $Q=\mathbb{R}^2$ and symmetric with respect to the
action of the group $G= \mathbb{R}$ of translations along the
$q^2$ coordinate (using the old language in Mechanics, $q^2$ is a
cyclic coordinate). The momentum map is $J_L=\dot{q}^1$. For a
frozen value $\mu = \dot{q}^1 $ of $J_L$, the quotient
$J^{-1}_L(\mu)/G_{\mu}$ can not be identified with $T(Q/G)\times
Q/G_\mu=T(Q/G)$ as the coordinates of the quotient space
$J^{-1}_L(\mu)/G_{\mu}$ are precisely $(q^1,\dot{q}^2)$ whereas
the coordinates of $T(Q/G)$ are $(q^1,\dot{q}^1)$. See \S
\ref{suy} below. The main result of this paper is the
generalization of Routh reduction to a reduction technique that
holds for arbitrary Lagrangians with the only requirement of $\mu$
being a regular value of $J_L$. The price paid for that is that
the reduction process will be carried out in the entire space $TQ$
and not only on $J^{-1}_L (\mu)$. The fact that the reduction takes
place in the whole space does not mean that we loose control of the
momentum, at the contrary, the reduction process will keep track of
the value of $J_L$ in the same spirit of the classical Routh reduction.

In order to relate the presented reduction technique to existing results on geometric Routh reduction, we first mention that the correspondence between Routh
reduction and Marsden-Weinstein reduction holds at a much more fundamental
level: Routh-reduction is Marsden-Weinstein reduction applied to
the Poincar\'e-Cartan symplectic structure on the tangent bundle (see~\cite{BC}).
This correspondence was then used to generalize Routh reduction to
Lagrangians that are quasi-invariant (up to a total time
derivative).

Here we wish to generalize Routh reduction from a different point
of view. In classical treatments on Routh reduction of a
Lagrangian system with cyclic coordinates, the reduced system is
again a Lagrangian system with a new Lagrangian which is called
the Routhian. It is this observation what we take as a starting
point in this paper: Routh reduction is a reduction technique that
preserves the Euler-Lagrange nature of the original system.

The different steps in the proposed reduction are best illustrated
by means of the schematic diagram in Fig.~\ref{fig:schema}. For
that purpose, we shall denote for now a Lagrangian system as a
couple $(TQ,L)$, $Q$ being the configuration manifold and $L$ a
function on $TQ$ called the Lagrangian. The first step is to
consider an {\em equivalent} Lagrangian system on $Q$, with
Lagrangian $R^\mu$. This new Lagrangian has the property that
solutions of the Euler-Lagrange equations of $(TQ,L)$ with
momentum $\mu$ are solutions to $(TQ,R^\mu)$ with {\em momentum}
0. For the sake of completeness we mention here that $(TQ,R^\mu)$
is {\em not conservative} and additional force terms should be
taken into account; and that the function $R^\mu$ is only $G_\mu$
invariant but, as will become clear, we can make sense to a
momentum map of $R^\mu$ taking values in $\lag^*$, the dual of the
Lie-algebra of $G$. This step is not new and was carried out, for
instance, in~\cite{jalna}. The Routh reduction technique described
in~\cite{jalna} is schematically presented in
Fig.~\ref{fig:schema} by the arrows $1$ and $2$, followed in
this order.

\begin{figure}[t]\centering
\includegraphics{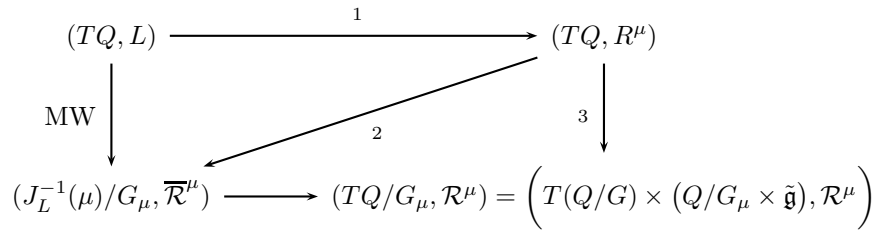}
\caption{Schematic drawing of reduction on a tangent
bundle}\label{fig:schema}
\end{figure}

We shall follow a different path: in our approach the second step
is to perform a reduction on the system $(TQ,R^\mu)$ that is
similar to Lagrange-Poincar\'e reduction, i.e. we reduce the
system to $(TQ/G_\mu, \Ro^\mu)$, where $\Ro^\mu$ is the quotient
of $R^\mu$ (represented by arrow $3$ in the above diagram). We
show that this step, in contrast to general Lagrange-Poincar\'e
reduction, preserves the Euler-Lagrange nature of the original
system. To explain this statement in more detail, we mention here
that $TQ/G_\mu$ can be identified with the fibred product of three
bundles over $Q/G$: $T(Q/G)$, $Q/G_\mu$ and $\tilde\lag$ (the
latter being the bundle associated to $\lag$). Roughly said, the
Routhian $\Ro^\mu$ depends on velocities $v_x$ in $T(Q/G)$ and on
points $y$ and $\tilde\xi$ in the fibres of $Q/G_\mu$ and
$\tilde\lag$ above $x$ respectively. When we say that
$\Ro^\mu(v_x,y,\tilde\xi)$ is a Lagrangian, we understand that the
variables $y$ and $\tilde\xi$ are to be interpreted as {\em
configuration} coordinates. This might be surprising, since
$\tilde\xi$ is the projection of the vertical part of a velocity
upstairs, i.e. a tangent vector in $TQ$.

It is immediately clear now that the Lagrangian $\Ro^\mu$ is {\em
singular}, i.e. it does not depend on the velocities of the
variables $y$ and $\tilde\xi$. This is another price we pay for
keeping the Euler-Lagrange nature of the system: we have to work
with singular Lagrangians. But the {\em singularity} of these
systems is of a specific type, which we call throughout this paper
{\em intrinsically constrained Lagrangian systems}. They are
studied in detail in Section~\ref{sec:lagsyst}. It is well known
that singular Lagrangians contain {\em constraints} on the
solutions to the associated Euler-Lagrange equations. It is shown
that one of these intrinsic constraints embedded in the system
$(TQ/G_\mu,\Ro^\mu)$ is a reduced version of the {\em  fixed
momentum condition} $J_L=\mu$. As a result of this Lagrangian
interpretation of the reduced system $(TQ/G_\mu,\Ro^\mu)$, we
shall define a presymplectic formulation and, in the most general
case, one can apply the presymplectic constraint
algorithm~\cite{gotaya,gotayb,gotayc} to find solutions to the
Euler-Lagrange equations.

Depending on the nature of these constraints we describe in more
detail two distinct cases where this reduced singular Lagrangian
system is equivalent to a variational problem on the space
$J^{-1}_L(\mu)/G_\mu$; the first case leads to standard Routh
reduction and the second case is new and, for instance, is
applicable to geodesics in general relativity where the metric is
invariant under a lightlike vector field.

Finally, we wish to mention that besides the importance of Routh
reduction for Mechanical systems by itself, the special features
of this process may shed light to other situations. In particular,
this might be the case for Field Theories. The covariant reduction
of these Lagrangian systems under the action of a group of
symmetries has been developed in recent years (see for example
\cite{CaRa}), though the Hamiltonian counterpart is much less
explored. A generalization of Routh techniques to this setting
could be of much interest to tackle some of the difficulties
encountered in this context (for instance, the lack of meaning of
the notion of fixing the value of the momentum map).

Throughout this paper manifolds are always assumed to be smooth
finite dimensional (Hausdorff, second countable, $C^\infty$) and
smooth always means of class $C^\infty$. We will often consider
the pull-back of a bundle, a function, a one- or two-form along a
map $f:M\to N$ between two manifolds. To reduce the notational
complexity in this paper, we sometimes denote the pull-backed
object with the same symbol as the object itself. It should be
clear from the context what is meant. Given two bundles
$\pi_1:M_1\to N$ and $\pi_2:M_2\to N$ over a manifold $N$, then we
often consider the fibred product bundle $M_1\times_NM_2$, or
simply $M_1\times M_2$ when no confusion is possible. Elements in
such a fibred product  are denoted as a couple $(m_1,m_2)$, with
$m_1\in M_1$ and $m_2\in M_2$ such that $\pi_1(m_1)=\pi_2(m_2)$.
Finally, velocities at a point $m\in M$ are typically denoted by
$v_m,w_m \in T_mM$. When we consider the tangent vector to a curve
$m:I\to M$ at a time $t$, we shall write $\dot m(t)$ for the
corresponding curve in $TM$.

\section{Lagrangian systems on fibred manifolds}\label{sec:lagsyst}

We start with recalling some standard concepts in Lagrangian
mechanics. Assume that $L$ is a Lagrangian function defined on the
tangent space of a manifold $M$.
\begin{definition} A Lagrangian system is a triple $(M,L,F)$ with $L$ a
function on $TM$ and $F$ a co-vector valued function $TM$, i.e.
$F:TM\to T^*M$, fibred over the identity. The map $F$ is called
the force term of the Lagrangian system. A Lagrangian system
$(M,L,F)$ is said to be conservative if $F=0$.
\end{definition}
In the above notion of conservative systems we assume that
conservative forces are always taken into account by means of
their potential energy in the Lagrangian function $L$.
\begin{definition}
A curve $m: I=[a,b] \to M$ is said to be a {\em critical} curve
for the Lagrangian system $(M,L,F)$ if for arbitrary variations
$\delta m$ with fixed endpoints $\delta m(a,b)=0$,
\[
\delta \int_I L(\dot m(t))dt +\int_I \langle F(\dot m(t)),\delta
m(t)\rangle dt=0
\]
holds.
\end{definition}
It is well-known that a critical curve satisfies
\[
{\cal EL}(L)(\ddot m(t)) +F(\dot m(t)) = 0,
\]
with ${\cal EL}$ the Euler-Lagrange operator ${\cal EL}(L):
T^{(2)}M \to T^*M$. When expressed in a local coordinate system
$(X^i)$ on $M$, we get
\[
{\cal EL}(L)(\ddot m(t))= \left(\fpd{L}{X^i}(X,\dot X) -
\frac{d}{dt}\left(\fpd{L}{\dot X^i}\right)(X,\dot X,\ddot X)
\right)dX^i.
\]
Thus a critical curve satisfies
\[
\frac{d}{dt}\left(\fpd{L}{\dot X^i}\right)-\fpd{L}{X^i}=
F_i,\mbox{ for } i=1,\ldots,\dim M.
\]
We now assume that the manifold $M$ is fibred over $N$, i.e. there
is a bundle map $\pi: M\to N$. We denote the kernel of $T\pi$ by
$V\pi\subset TM$ and is called the bundle of vertical tangent
vectors. A cotangent vector $\alpha_m\in T^*_mM$ can be restricted
to $V_m\pi$, and such a restriction will be denoted by
$\alpha_m^v: V_m\pi \to \R$. Moreover, we denote the product
bundle $TN\times_N M$ by $T_MN$. Elements in this bundle are
denoted by $(v_n,m)$, where $v_n\in T_nN$ and $m \in
M_n=\pi^{-1}(n)$. The map $(T\pi,\tau_M):TM\to T_MN; v_m\mapsto
(T\pi(v_m),m)$ is denoted by $p_1$.

\begin{definition}
  An intrinsically constrained Lagrangian system is a triple
  $(\pi:M\to N,L,F)$, with $L$ a function on $T_MN$ and $F$ the
  force term $F:TM\to T^*M$ fibred over the identity on $M$. A
  curve $m(t)\in M$ is said to be a {\em critical curve}  of the
  intrinsically constrained Lagrangian system if it is a critical
  curve of the associated Lagrangian system $(M,p_1^*L,F)$.
\end{definition}

It should be clear that the Lagrangian system $(M,p_1^*L, F)$ is
singular since $p_1^*L$ is independent of the velocities in the
fibre coordinates. We now focus on the Euler-Lagrange equations
for the Lagrangian system $(M,p_1^*L,F)$ associated to the
intrinsically constrained system $(\pi:M\to N,L,F)$ (from now on,
we denote $L$ and $p_1^*L$ by the same symbol). For that purpose,
let $(x^i,y^a)$ be a coordinate system on $M$ that is adapted to
the fibration, i.e. $(x^i; i=1,\ldots,\dim N)$ are coordinates on
$N$ and $(y^a; a=1,\ldots ,\dim M-\dim N)$ are coordinates for the
typical fibre of $\pi$. It is not hard to see that the
Euler-Lagrange operator for an intrinsically constrained
Lagrangian system, when restricted to vertical directions becomes
\[
{\cal EL}(L)^v = \fpd{L}{y^a}dy^a.
\]
A critical curve $m(t)=(x^i(t),y^a(t))$ of the Lagrangian system
on $\pi:M\to N$ then satisfies the (possibly nonholonomic)
condition
\begin{equation}\label{eq:constraint}
\fpd{L}{y^a}(x,\dot x,y)+F_a(x,y,\dot x,\dot y)=0, \mbox{ for }
a=1,\ldots,\dim M-\dim N.
\end{equation}

The appearance of this {\em constraint} on the set of critical
curves justifies the previous definitions. We now specialize to
three distinct cases for constrained Lagrangian systems. We focus
on a local treatment.

\subsection{Gyroscopic forces}\label{sec:velcon}

We study here the specific case where the non-conservative force
term $F$ is of the form $F(v_m) = -i_{v_m} \beta_m$ with $\beta$ a
closed 2-form on $M$. In this case we say that $F$ is a {\em
gyroscopic} force term associated to $\beta$. As will become clear
further on, Routh reduction of a Lagrangian system with
non-abelian symmetry is {\em in general} a Lagrangian system
subjected to gyroscopic forces. This type of forces therefore
deserves special attention. Let $v_m=(x^i,y^a, \dot x^i,\dot y^a)$
in a coordinate system $(x^i,y^a)$, then $F$ assumes the following
form
\[
F(v_m) = \left(\beta_{ij} \dot x^j  + \beta_{ia}\dot
y^a\right)dx^i + \left(\beta_{ab}\dot y^b-\beta_{ia}\dot
x^i\right)dy^a.
\]
We say that the gyroscopic force term is {\em regular} if
$\beta_{ab}$ is nondegenerate at all points in $M$, i.e. if
$\beta$ is a nondegenerate 2-form when restricted to $V\pi$. For
gyroscopic forces, the intrinsic constraint~(\ref{eq:constraint})
becomes:
\[
\fpd{L}{y^a}  + \left(\beta_{ab}\dot y^b-\beta_{ia}\dot
x^i\right)=0,
\]
and if $\beta$ is regular, it should be clear that this constraint
can be brought into the form
\[
\dot y^a = - \left((\beta|_{V\pi})^{-1}\right)^{ab}
\left(\fpd{L}{y^b}  -\beta_{ib}\dot x^i\right).
\]

An interesting property of intrinsically constrained Lagrangian
systems with a gyroscopic force term is that they admit a
presymplectic formulation on the manifold $T_M N$. To show this,
we introduce the following maps:
\begin{enumerate}
\item \label{def:legmap} the Legendre map $\F_{1} L: T_MN\to
T_M^*N$, defined by \[\langle \F_{1} L (v_n,m) , (w_n,m)\rangle =
\left.\frac{d}{d\epsilon} \right|_{\epsilon=0} L (v_n+\epsilon
w_n,m),\] for $(v_n,m),(w_n,m)\in T_MN$ arbitrary; \item the
energy $E_L$ of the Lagrangian $L$ is the function on $T_MN$
defined by \[E_L = \langle \F_1 L(v_n,m),(v_n,m)\rangle -L(v_n,m);\]
\item the projections $\pi_{1}: T_MN \to TN$ and $\pi_{2}: T_MN
\to M$, and their analogues on the cotangent bundle:
$\overline\pi_1: T_M^*N\to T^*N$ and $\overline\pi_2: T_M^*N \to
M$.
\end{enumerate}
We denote the Liouville form on $T^*N$ by $\theta_N$ and the
associated symplectic 2-form $\omega_N=d\theta_N$. The following
result is then easily proven in a local coordinate system.
\begin{proposition}\label{prop:presymp}
Let $m(t)$ denote a critical curve of the intrinsically
constrained Lagrangian system $(M,L,F)$ with gyroscopic force term
associated to $\beta$. Then $\gamma(t)=(\dot n(t),m(t))\in T_MN$,
with $n(t)=\pi(m(t))$, satisfies the following presymplectic
equation
\[
\left(i_{\dot \gamma (t)
}\big((\overline\pi_1\circ\F_1L)^*\omega_N + \pi^*_2\beta\big) =
-dE_L\right)|_{\gamma}.
\]
The converse also holds: any curve $\gamma(t)\in T_MN$ of the form
$\gamma(t)=(\dot n(t),m(t))$ with $n(t)=\pi(m(t))$ and that solves
the above presymplectic equation, projects to a critical curve
$m(t)=\pi_2(\gamma(t))$ of $(M,L,F)$.
\end{proposition}
If (i) $\F_1L:T_MN\to T_M^*N$ is a diffeomorphism and if (ii)
$\beta$ determines a regular gyroscopic force term, then the
2-form $\big((\overline\pi_1\circ\F_1L)^*\omega_N +
\pi^*_2\beta\big)$ is symplectic. In~\cite{BC} the symplectic
structure of a Routh reduced system is proven to be of this type
and it is identified with the the Marsden-Weinstein reduced
symplectic structure of a Lagrangian system with symmetry. Without
the above regularity assumptions (i) and (ii), the presymplectic
system may only have solutions on a submanifold of $T_MN$. A next
step would be to find these submanifolds, however, this is not the
scope of this paper and we refer the reader to the presymplectic
constraint algorithm~\cite{gotaya,gotayb}.

We end this paragraph with a slightly more general formulation of
the above proposition. In general, the intrinsically constrained
Lagrangian systems obtained from Routh reduction will have a force
term $F=F_1+F_2$ that consists of two parts: a gyroscopic force
term $F_1$ associated to $\beta$ and a force term $F_2=
T^*\pi\circ\hat F\circ p_1$, with $\hat F:T_MN\to T^*N$. Again
from local coordinate expressions, we conclude that critical
curves of $(\pi:M\to N,L,F)$ solve the presymplectic equation (and
vice versa)
\[
\left(i_{\dot \gamma (t)
}\big((\overline\pi_1\circ\F_1L)^*\omega_N + \pi^*_2\beta\big)  =
-dE_L+ \pi_1^* \hat F\right)|_{\gamma},
  \]
with $\pi_1^* \hat F:T_MN \to T^*(T_MN)$ such that $\langle
\pi_1^*\hat F(v_n,m), V_{(v_n,m)}\rangle =\langle \hat F(v_n,m),
T\pi_1(V_{(v_n,m)})\rangle$, with $(v_n,m)\in T_MN$ and
$V_{(v_n,m)}\in T_{(v_n,m)}(T_MN)$ arbitrary.

\subsection{Regular configuration constraints}\label{sec:regconf}

Throughout this paragraph we assume that the force term $F$ is of
the form $F=T^*\pi \circ \hat F\circ p_1$ with $\hat{F}\colon T_M
N \to T^* N$ as above, that is, $F(v_m) = T^*_m \pi(\hat F(v_n,m))$,
with $v_n=T\pi(v_m)$. Then the constraint
equation~(\ref{eq:constraint}) reduces to
\begin{equation}\label{eq:con2}
\pd L/\pd y^a(x,\dot x,y) = 0.
\end{equation}
Intrinsic constraints of this form are called {\em configuration
constraints} and they are said to be {\em regular} if a map
$\gamma:TN\to M$, with $\pi\circ\gamma=\tau_N$ exists such that,
in local coordinates, Eq.~(\ref{eq:con2}) is equivalent to
$y^a=\gamma^a(x,\dot x)$. As we show here, in this case critical
curves of the Lagrangian system on $M$ are in a one to one
correspondence to critical curves of a new Lagrangian system on
$N$. The local existence of such a map $\gamma$ is guaranteed if
the (vertical) Hessian of the Lagrangian is non-degenerate, i.e.
if
\[
\mbox{det}\fpd{^2L}{y^a\partial y^b}\neq 0.
\]
Consider the Lagrangian system $(N,L',F')$ on $N$, with
\[
L' = L(v_n,\gamma(v_n)) \mbox{ and } \langle F'(v_n), w_n\rangle
=\langle \hat F(v_n,\gamma(v_n)),w_n\rangle.
\]
Let $m(t)=(x^i(t),y^a(t))$ be a curve in $M$. It is easily seen
that the condition for $m$ to be critical w.r.t $(M\to N, L, F)$
is equivalent to $n(t)=\pi(m(t))$ being a critical curve of
$(N,L',F')$.  More precisely,
\begin{eqnarray*}
&&\frac{d}{dt}\left(\fpd{L}{\dot x^i}(x,\dot x,y)\right)
=\fpd{L}{x^i}
(x,\dot x,y) + F_i(x,\dot x,y), \mbox{ and}\\
&&y^a = \gamma^a(x,\dot x).
\end{eqnarray*}
is equivalent to
\[
\frac{d}{dt}\left(\fpd{L'}{\dot x^i}(x,\dot x)\right)
=\fpd{L'}{x^i}(x,\dot x) + F_i(x,\gamma(x,\dot x),\dot x)
\]
The correspondence is made explicit by the map $\gamma$:
$\gamma(n(t))=m(t)$. In this sense the intrinsically constrained Lagrangian system
$(M\to N,L,F)$ is `equivalent' to the Lagrangian system
$(N,L',F')$ on $N$.

The relation between Routh reduced systems and intrinsically
constrained systems can now be made more explicit, and for this purpose we go ahead of some definitions and notations. We believe however that the reader familiar with geometric Routh reduction will benefit from the following. In general, a
Routh reduced Lagrangian system will be an intrinsically
constrained system on $M=(Q/G_\mu\times \tilde\lag)$ fibred over
$N=Q/G$, where $Q$ is a manifold on on which a Lie group $G$ acts freely
and properly. The bundle $M$ is a product bundle, and the Routhian
is a Lagrangian on $TM$ independent of the velocities in the
fibres of $M$. As will be made more precise in the following, this
leads to constraints on the set of critical curves and due to the
fact that $M$ is a product bundle we will distinguish two
constraints. The constraints arising from the bundle $Q/G_\mu\to
Q/G$ will always belong to the regular gyroscopic class. The
constraints arising from the bundle $\tilde\lag\to Q/G$ will be
configuration constraints. If these configuration constraints are
regular, then we retrieve the standard Routh reduction procedure.
The irregular case will sometimes lead to a variational problem on
$N$ as is described in the next paragraph. In general, a Routh
reduced system is associated to a presymplectic structure on
$T_MN$, and its solutions can be constructed by means of the
presymplectic constraint algorithm.

\subsection{Linear constraints and Lagrangian multipliers}
Throughout this paragraph we describe a special case of
configuration constraints for a Lagrangian system on $M$ with the
property that its critical curves correspond to the solutions of a
variational problem on $N$. We assume that (i) $\pi:M\to N$ is a
linear fibration and (ii) the force term $F$ is of the form $F=
T^*\pi\circ \hat F \circ T\pi$, with $\hat F:TN\to T^*N$.

We say that the intrinsically constrained system $(M,L,F)$ is {\em
linear} if $L$ is of the form
\[
L(v_n,m) = L_0(v_n) + \langle \alpha(v_n),m\rangle.
\]
with $\alpha: TN\to M^*$, fibred over the identity and $L_0$ a
function on $TN$. The constraint equation~(\ref{eq:constraint}) for critical curves $m(t)=(x^i(t),y^a(t))$ now becomes
\[
\fpd{L}{y^a}(x,\dot x,y) = \alpha_a(x,\dot x) = 0.
\]

Assume that $0$ is a regular value of $\alpha$, and consider the
submanifold $C$ of $TN$ defined by $v_n\in C$ if and only if
$\alpha(v_n)=0$. The following proposition states that critical
curves of $(\pi:M\to N,L,F)$ coincide with critical curves of the
variational problem on $N$ with Lagrangian $L_0$ constrained to
the submanifold $C$ in $TN$. Therefore, we can say that the
initial Lagrangian system on $M$ admits a variational
interpretation on the base manifold $N$, independent of the total
space $M$. When working in a coordinate system, one might
interpret the fibre coordinates $y^a$ as Lagrangian multipliers
for the constrained variational problem on the base manifold.
\begin{definition}
A critical curve $p(t)$ of a Lagrangian system $(P,L',F')$ on a
manifold $P$ {\em constrained} to a submanifold $C'$ of $TP$ is a
curve satisfying
\[
\delta \int_I L'(\dot p(t))dt +\int_I \langle F'(\dot p(t)),\delta
p(t)\rangle dt=0,
\]
for any variation $\delta p$ with fixed endpoints originating from
a deformation $p_\e(t)$ of $p(t)$ which is entirely contained in
$C'$, i.e. $\dot p_\e(t)\in C'$.
\end{definition}

\begin{proposition}\label{prop:lagcon}
Given a linear intrinsically constrained Lagrangian system,
$(\pi:M\to N,L=L_0+\langle\alpha,m\rangle,F)$ on $M$. Then the
critical curves of $(\pi:M\to N,L,F)$ on $M$ project under $\pi$
onto critical curves of the Lagrangian system $(N,L_0,\hat F)$ on
$N$ constrained to $C$.
\end{proposition}
\begin{proof}
The proposition is easily proven by noting that critical curves of
$(\pi:M\to N,L,F)$ satisfy
\[
\delta \int_I L(\dot m(t))dt +\int_I \langle \hat F(\dot
n(t)),\delta n(t)\rangle dt=0
\]
for arbitrary variations with vanishing endpoints in $M$ and
project onto curves in $C$. Among all these variations, we can now
consider variations originating from deformations $m_\e(t)$ of
$m(t)$ that project onto deformations $n_\e(t)$ of $n(t)$ that are
contained in $C$. In that case the first term in the latter
equation is precisely
\[
\left.\frac{d}{d\e}\right|_\e \int_I L(\dot m_\e(t))dt =
\left.\frac{d}{d\e}\right|_\e\int_I L_0(\dot n_\e(t))dt,
\]
and we thus find that $n(t)$ is a critical curve of $(N,L_0,\hat
F)$ constrained to $C$.
\end{proof}

\subsection{Connections}

We now return to the general setting of intrinsically constrained
Lagrangian systems and conclude this section with a proposition
which describes how to split up the Euler-Lagrange equations
making use of a chosen connection on the fibration $\pi:M\to N$. A
connection on $M$ is a direct sum decomposition of $TM$, i.e. a
distribution $H$ on $M$ such that $TM=H\oplus V\pi$, and the
associated map $h:T_MN \to H$ is called the horizontal lift. Using
this connection we are now able to say that a variation $\delta
m(t)$ of $m(t)$ is {\em horizontal} if $\delta m(t) \in H$, or if
equivalently $\delta m(t) = h(m(t),\delta n(t))$, where $\delta
n(t)$ is the corresponding variation of $n(t)=\pi(m(t))$. We
denote the horizontal and vertical components of an arbitrary
tangent vector $w$ by $w^h$ and $w^v$ respectively. The
restriction of a cotangent vector $\alpha$  to the horizontal
distribution is denoted by $\alpha^h$.
\begin{proposition}\label{splittingel}
The Euler-Lagrange operator of an arbitrary Lagrangian system
$(M,L,F)$ on $M$ is completely determined if one only considers
horizontal and vertical variations, i.e. if $m(t)$ is a critical
curve of $(M,L,F)$ then it satisfies:
\begin{eqnarray*}
&&\delta \int_I L(\dot m(t))dt + \int_I \langle F(\dot m(t)),\delta
m(t)\rangle dt= 0 \mbox{ with } \delta m^v(t) =0,\\
&&\delta \int_I L(\dot m(t))dt + \int_I \langle F(\dot m(t)),\delta
m(t)\rangle dt= 0 \mbox{ with } \delta m^h(t) =0.
\end{eqnarray*}
The first and second equations are equivalent to, respectively,
\begin{eqnarray*}
&&{\cal EL}(\ddot m(t))^h +F(\dot m(t))^h=0\in T^*_nN\cong H_m,\\
&&{\cal EL}(\ddot m(t))^v + F(\dot m(t))^v=0 \in V_m^*\pi.
\end{eqnarray*}
\end{proposition}
Assume an adapted coordinate chart is chosen and that the
horizontal distribution is spanned by the following set of vector
fields
\[
\fpd{}{x^i} - \Gamma^a_i(x,y)\fpd{}{y^a},
\]
with $\Gamma^a_i$ the connection coefficients, then the horizontal
and vertical parts of the Euler-Lagrange equations take the
following form
\begin{eqnarray*}
&&\fpd{L}{x^i} -\frac{d}{dt} \left(\fpd{L}{\dot x^i}\right) +F_i-
\left( \fpd{L}{y^a}+F_a\right)\Gamma^a_i =0,\\
&&\fpd{L}{y^a}+ F_a =0.
\end{eqnarray*}
In fact, by substituting the second in the first equation, we may
conclude that these equations are independent of the connection
coefficients.

\section{Lagrangian systems with symmetry}\label{sec:lagsym}

We consider a Lagrangian system $(Q,L,F)$ on a manifold $Q$ and we
assume that a Lie group $G$ freely and properly acts on $Q$. This ensures that $Q/G$
has a manifold structure and that $\pi :Q\to Q/G$ is a principal fibre bundle. The
(right) action is denoted by $\Psi:G\times Q \to Q: (g,q)\mapsto
\Psi(q,g)=qg$. The associated maps $\sigma_q: G\to Q$ and $R_g:
Q\to Q$ are defined by
\[
qg=\sigma_q(g) = R_g(q).
\]
For the sake of simplicity, the infinitesimal action induced by
$\sigma$ will be also denoted by $\sigma$: $\sigma_q(\xi) =
T\sigma_q(\xi)$, with $\xi\in\mathfrak{g}$ a Lie-algebra element
of $G$. The vector field $\xi_Q$ on $Q$ is defined by $\xi_Q(q) =
\sigma_q(\xi)$. We can consider the lifted action of $G$ on $TQ$,
and using the previous definition on $TQ$, we shall write
$\xi_{TQ}$ to denote the symmetry vector field associated with the
Lie algebra element $\xi\in\lag$. Throughout this section we shall
extensively rely on the choice of a principal connection on $Q$.
This is a $\lag$-valued one-form $\omega$ on $Q$ satisfying
$\omega(\xi_Q)=\xi$ and which is equivariant in the sense that,
for arbitrary $g\in G$,
\[
R^*_g\omega= Ad_{g^{-1}}\cdot \omega,
\]
where $Ad$ denotes the adjoint action on $\lag$. We assume that
the notions of horizontal lift, associated bundles and covariant
derivatives, etc are well-known (see~\cite{koba}). One typically
defines the momentum map associated with the Lagrangian $L$ in the
following way.

\begin{definition}
The momentum map $J_L: TQ \to \lag^*$ is the fibre derivative of
$L$ restricted to vertical directions:
\[
J_L(v_q)(\xi)=\left.\frac{d}{d\e}\right|_{\e=0} L(v_q+\e
\xi_Q(q)).
\]
\end{definition}

It is well-known that if the Lagrangian is invariant under the
(prolonged) action of $G$ on $TQ$, then (i) $J_L$ is an
equivariant function, i.e. $R_g^*J_L = \mbox{Ad}^*_{g}\cdot J_L$
and (ii) $J_L$ is a first integral of the critical curves of the
{\em conservative} Lagrangian system $(Q,L,F=0)$.
\begin{definition}\label{def:invarfor}
A force term $F: TQ\to T^*Q$ is said to be invariant under the
action of $G$ (or simply $G$-invariant) if the following two
conditions are satisfied:
\begin{enumerate}
\item $\langle F(v_qg), w_qg\rangle =\langle F(v_q),w_q\rangle$
\label{disone} \item $\langle F(v_q), \xi_Q(q)\rangle =0$,
\label{distwo}
\end{enumerate}
for $v_q,w_q\in T_qQ$ and $\xi \in \lag$ arbitrary.

A Lagrangian system $(Q,L,F)$ is $G$-invariant if the Lagrangian
$L$ is an invariant function under the action of $G$ and if the
force term $F$ is $G$-invariant.
\end{definition}
A straightforward consequence of the definition of a $G$-invariant
force term $F$ is that $F$ is reducible to a map $f$ between the
quotient spaces $TQ/G$ and $T^*(Q/G)$. Moreover, the second
condition in definition~\ref{def:invarfor} guarantees that the
momentum map $J_L$ is conserved along critical curves of a
$G$-invariant Lagrangian system $(Q, L,F)$ (it is sufficient, but
not necessary). In the following proposition we study the
necessary conditions for a Lagrangian system $(Q,L,F)$ such that
$J_L$ is conserved along its critical curves. This generality will
be of importance further on when considering the Routhian and its
momentum map, both associated with an invariant Lagrangian.
\begin{proposition}\label{prop:cm}
The momentum map for a (not necessarily $G$-invariant) Lagrangian
system $(Q,L,F)$ is conserved if and only if the force term is
such that $\langle dL, \xi_{TQ}\rangle=-\langle F,\xi_Q\rangle$
for all $\xi\in\lag$ along the critical curves. If this holds we
say that the Lagrangian system $(Q,L,F)$ has conserved momentum
$J_L$ w.r.t. the action of $G$.
\end{proposition}

\begin{proof}
Assume that $q(t)$ is a critical curve of the Lagrangian system
$(L,F)$. We now compute the derivative of the momentum map along
the trajectory and proof that it is equal to zero. Let $(q^i)$ be
a local coordinate system on $Q$ and let $\xi\in\lag$ be chosen
arbitrarily, then
\begin{eqnarray*}
\left.\frac{d}{dt} \right|_t\left(J_L(\dot q(t))(\xi)\right) &=&
\left. \frac{d}{dt} \right|_t\left.\frac{d}{d\e}
\right|_{\e=0}L(\dot q(t)+\e \xi_Q(q(t)))
 = \left.\frac{d}{dt} \right|_t \left(\fpd{L}{\dot q^i} \xi_Q^i\right)\\
&=& \frac{d}{dt}\left(\fpd{L}{\dot q^i}\right)\xi_Q^i + \fpd{L}{\dot q^i} \fpd{\xi_Q^i}{q^j}\dot q^j\\
&=& \fpd{L}{q^i} \xi_Q^i + F_i\xi_Q^i + \fpd{L}{\dot q^i} \fpd{\xi_Q^i}{q^j}\dot q^j\\
&=& \xi_{TQ}(L)(\dot q(t)) + \langle F(\dot
q(t)),\xi_Q(q(t))\rangle = 0.
\end{eqnarray*}
\end{proof}
Note that, if $L$ is invariant then the momentum map is conserved
along critical curves if $F$ satisfies condition $(2)$ from
definition~\ref{def:invarfor}.

The first step in our approach to Routh reduction of a
$G$-invariant Lagrangian system $(Q,L,F)$ consists of defining an
equivalent Lagrangian system on $Q$ whose Lagrangian is only
invariant under the the action of a subgroup of $G$ although it
has conserved momentum w.r.t. the action of $G$. This observation
should justify the above definitions.

\subsection{General constructions on the quotient spaces}
Assume throughout the remaining of this paper that a {\em regular
value} $\mu\in \lag^*$ for the momentum map of a $G$-invariant
Lagrangian system $(Q,L,F)$ is chosen. We denote the group $G_\mu$
as the isotropy subgroup of $G$ of $\mu$ under the coadjoint
action of $G$ on $\lag^*$.

This section describes the structure of the manifold $TQ/G_\mu$ on
which the Routhian lives and how variations of a curve in $Q$
behave when projected onto the quotient $TQ/G_\mu$. Many of the
properties explained below can be found in more detail
in~\cite{CMR01}. We start with fixing notations and defining
several projections and bundles appearing in the chosen framework.
Recall that the bundle $Q/G_\mu\to Q/G$ is an associated bundle of
the principal fibre bundle $Q\to Q/G$, whose typical fibre is the
orbit ${\cal O}_\mu$ of $\mu $ under the adjoint action on
$\lag^*$.
\begin{definition}\label{def:bundles}
The projections onto the different quotient spaces are denoted by
\begin{eqnarray*}
\pi:Q\to Q/G,& &\pi_\mu:Q\to Q/G_\mu,\\
\overline \pi_{\mu}:Q/G_\mu \to Q/G, & & \tilde \pi:\tilde\lag :=
(Q\times \lag)/G \to Q/G.
\end{eqnarray*}
\end{definition}
A point in $Q$, $Q/G_\mu$ and $Q/G$ is typically denoted by $q$,
$y$ and $x$ respectively. Similarly, tangent vectors in $T_qQ$,
$T_y(Q/G_\mu)$ and $T_x(Q/G)$ are denoted by $v_q$, $v_y$ and
$v_x$ respectively. Further, elements in the bundle $\tilde\lag\to
Q/G$ associated to the Lie-algebra $\lag$ are denoted by
$\tilde\xi=[q,\xi]_G$, with $\xi\in\lag$, $q\in Q$.

If a principal connection $\omega$ on $\pi:Q\to Q/G$ is chosen,
then the horizontal lift of a tangent vector in $T(Q/G)$ to the
point $q\in\pi^{-1}(x)$ is denoted by $(v_x)_q^h$. The horizontal
lift of a curve $x(t)$ from $Q/G$ to $Q$ is denoted by $q_h(t)$

According to our previous notation conventions, an element in the
product bundle $T(Q/G)\times {Q/G_\mu}\times \tilde \lag\to Q/G$
is denoted by $(v_x,y, \tilde\xi)$, where $v_x\in T_x(Q/G), y\in
Q/G_\mu$ and $\tilde\xi\in\tilde\lag$ satisfy
$\overline\pi_\mu(y)=\tilde \pi(\tilde\xi)=x\in Q/G$.
\begin{proposition}
\label{poron} The manifold $TQ/G_\mu$ is isomorphic to
$T(Q/G)\times Q/G_\mu\times \tilde \lag$. The isomorphism can be
made explicit with the choice of a principal connection $\omega$
on $Q$.
\end{proposition}
\begin{proof}
Assume that a connection $\omega$ on $Q$ is chosen. Let $[v_q ]_{G_\mu}\in TQ/G_\mu$ be the orbit of an arbitrary tangent vector $v_q$ in $TQ$. We now define a map
\begin{equation}
\label{poon} \phi_\omega: TQ/G_\mu \to T(Q/G)\times Q/G_\mu \times
\tilde \lag,
\end{equation}
together with its inverse $\psi_\omega$ as
\begin{eqnarray*}
&&\phi_\omega([v_q]_{G_\mu})=
(T\pi(v_q),\pi_\mu(q),[q,\omega(q)(v_q)]_G),\\
&&\psi_\omega(v_x , y,\tilde\xi) =
[(v_x)_q^h+\sigma_q(\xi)]_{G_\mu}, \ \mbox{with } y = [q]_{G_\mu}
\mbox{ and } \tilde\xi=[q,\xi]_G.
\end{eqnarray*}
It is standard to show that these maps are smooth and do not depend on the choice
of an element in the orbits under consideration.
\end{proof}

Using similar arguments as before, one can show that $TQ/G\cong
T(Q/G)\times \tilde\lag$. Taking into account these isomorphisms,
the natural projection $\lambda _\mu \colon TQ/G _\mu \to TQ/G$
simply reads
\[
\lambda_\mu(v_x, y,\tilde \xi)= (v_x,\tilde \xi).
\]
\begin{definition}
An invariant Lagrangian function $L$ on $TQ$ reduces to a function
$l$ on $TQ/G$, defined by
\[
l(v_x,\tilde\xi)=L((v_x)^h_q+\sigma_q(\xi)), \mbox{ with }
q\in\pi^{-1}(x),[q,\xi]_G=\tilde\xi.
\]
The function $\lambda^*_\mu l$, defined on $TQ/G_\mu$ is denoted
by the same symbol.

The associated momentum map $J_L:TQ \to \lag^*$, reduces to a
bundle  map $j_l: TQ/G \to \tilde\lag^*$, fibred over the identity
on $Q/G$ and defined by
\[
\langle j_l(v_x, \tilde\xi),\tilde \eta\rangle =\langle
J_L((v_x)^h_q+\sigma_q(\xi)), \eta\rangle, \mbox{ with }
q\in\pi^{-1}(x),[q,\xi]_G=\tilde\xi, [q,\eta]_G=\tilde \eta.
\]
\end{definition}

The reduced map $j_l$ can alternatively be defined from $l$ directly. Let
$\mathbb{F}_{\tilde{\xi}}l \colon
T(Q/G)\times\tilde{\mathfrak{g}}\to T(Q/G)\times
\tilde{\mathfrak{g}}^*$ be the fiber derivative of $l$ with
respect to the (linear) factor $\tilde{\mathfrak{g}}$.
\begin{proposition}
The maps $j_l$ and $l$ are related by:
\[
\langle \F_{{\tilde\xi}} l(v_x, \tilde \xi),(v_x,\tilde \eta)
\rangle = \langle j_l(v_x,\tilde \xi),\tilde \eta\rangle.
\]
\end{proposition}
\begin{proof}Let $(v_x,\tilde\xi), (v_x,\tilde\eta) \in T(Q/G)
\times\tilde\lag$ be arbitrary. Fix a point $q\in Q$ projecting to
$x$ such that $[q,\xi]_G=\tilde \xi$ and $[q,\eta]_G=\tilde \eta$;
then
\begin{eqnarray*}
\langle \F_{{\tilde\xi}} l(v_x,\tilde \xi),(v_x,\tilde \eta)
\rangle &:=& \left.\frac{d}{d\e}\right|_{\epsilon =0} l(v_x, \tilde
\xi +\e\tilde\eta)=\left.\frac{d}{d\e}\right|_{\epsilon =0}
L((v_x)^h_q+\sigma_q(\xi+\e\eta))\\
&=& \langle J_L((v_x)^h_q+\sigma_q(\xi)),\eta \rangle = \langle
j_l(v_x,\tilde\xi),\tilde \eta\rangle.
\end{eqnarray*}
\end{proof}
The element $\mu \in \lag^*$ defines a map $\tilde \mu:Q/G_\mu\to
\tilde \lag^*$ such that
 \begin{equation}
 \label{chuchu}
 \langle\tilde \mu(y),\tilde \xi\rangle= \langle\mu,\xi\rangle,
 \mbox{ with } [q]_{G_\mu}=y \mbox{ and } [q,\xi]_G = \tilde \xi.
 \end{equation}
Due to the equivariance of $J_L$, the level set $J^{-1}_L(\mu)$ is
$G_\mu$ invariant and reduces to a subset of $TQ/G_\mu$.
\begin{proposition}
The quotient manifold $J_L^{-1}(\mu)/G_\mu$ is identified with the subset of points $(v_x,y,\tilde\xi)$ in $TQ/G_\mu$
satisfying $j_l(v_x,\tilde\xi) = \tilde\mu(y)$.
\end{proposition}
In classical Routhian reduction, an essential condition for
reducibility of an invariant Lagrangian system is that this
submanifold $J^{-1}_L(\mu)/G_\mu$ is precisely $T(Q/G)\times
Q/G_\mu$. This situation is the subject of the next session.

\subsection{The regular case}\label{suy}

\begin{definition}
An invariant Lagrangian $L$ is said to be $G$-regular if the map
$\F_{\tilde\xi} l$ is a diffeomorphism between
$T(Q/G)\times\tilde\lag$ and $T(Q/G)\times\tilde\lag^*$.
\end{definition}

We shall write the inverse of this map as $\kappa_l:
T(Q/G)\times\tilde\lag^*\to T(Q/G)\times\tilde\lag$.

\begin{proposition}
\label{porron} If $L$ is $G$-invariant and $G$-regular then the
mapping
\begin{eqnarray*}
J^{-1}_L(\mu ) / G_\mu & \to & T(Q/G) \times Q/ G_\mu\\
{[ v_q ]_{G_\mu}} & \mapsto & (T\pi (v_q), [q]_{G_\mu})
\end{eqnarray*}
is a diffeomorphism.
\end{proposition}

\begin{proof}
We provide the inverse of the mapping of the statement. Given a
point $(v_x , y) \in T(Q/G)\times Q/ G_{\mu}$, and let
$\tilde\xi\in\tilde\lag$ be such that
\[
(v_x,\tilde\xi) = \kappa_l(v_x,\tilde\mu(y)).
\]
Fix a representant $q$ for $y$, i.e. $y=[q]_{G_\mu}$, and let
$\xi$ be such that $\tilde\xi=[q,\xi]_G$. Then the point $v_q=((v_x)^h
_q + \xi_Q (q))\in TQ$ belongs to the set $J^{-1}_L(\mu)$. Indeed,
with $\tilde\eta=[q,\eta]_G$ arbitrary,
\[
J_L((v_x)^h _q + \sigma_q(\xi))(\eta)= \langle
j_l(v_x,\tilde\xi),\tilde\eta\rangle=\langle
\tilde\mu(y),\tilde\eta\rangle = \langle \mu ,\eta \rangle.
\]
It only remains to show that this definition is independent of the
chosen representative $q$ at $y$. Due to the equivariance of $J_L$
this is standard. Finally, with $(v_x,y)$ we can now define an element in $J_L^{-1}(\mu)/G_\mu$ determined as the orbit of $v_q$. The mapping from $T(Q/G)\times Q/ G_{\mu}$ to $J^{-1}_L(\mu)$ obtained in this way is the inverse of the map defined in the statement of the proposition.
\end{proof}

\begin{remark}
Note that the mapping given in Proposition \ref{porron} is
precisely the first two components of the diffeomorphism
(\ref{poon}) in Proposition \ref{poron}. In other words, the
$G$-regularity shows that the submanifold $J^{-1}_L(\mu)/G_{\mu}$
in $T(Q/G)\times Q/G_{\mu} \times \tilde\lag$ can be given as the
set of points $(v_x , y , \tilde{\xi})$ with $(v_x,y)\in
T(Q/G)\times Q/G_\mu$ arbitrary and with $\tilde\xi$ such that
$(v_x,\tilde\xi)=\kappa_l(v_x,\tilde\mu(y))$.

With this in mind, it follows that the converse of Proposition~\ref{porron} is also valid: if for any $\mu\in\lag^*$ the map $J_L^{-1}(\mu)/G_\mu \to T(Q/G)\times Q/G_\mu$ is a diffeomorphism, then $L$ is $G$-regular. This shows that our definition of $G$-regular Lagrangians is independent of the connection $\omega$.
\end{remark}

\begin{remark}
Let $\rho$ be a $G$-invariant Riemannian metric on $Q$ and $V \in
C^\infty (Q)$ a $G$-invariant potential energy. The Lagrangian $L(v_q)=\frac12 \rho(v_q , v_q)-
V(q)$ (kinetic minus potential) is $G$-regular (see, for
example, \cite{marsdenrouth}).

Note however that a Lagrangian can be both hyperregular (i.e., its
Legendre transformation is a diffeomorphism) and $G$-invariant but
not $G$-regular. Recall the example from the introduction $L=
(\dot q ^1)^2 + \dot q^1 \dot q ^2 - V(q^1)$ with $V\in C^\infty
(\mathbb{R})$, $Q=\mathbb{R}^2$ and $G=\mathbb{R}$ acting by
translations in the $q^2$ coordinate. Fix a connection, say $\omega=dq^2$, and let $\tilde\xi=(q^1,\xi), \tilde\eta=(q^1,\eta)$ be arbitrary in $\tilde\lag=Q/G\times\R$. Then the reduced Lagrangian $l:T(Q/G)\times \tilde{\mathfrak{g}}\to \mathbb{R}$ is $l(q^1,\dot{q}^1,\xi)
=(\dot{q}^1)^2+\dot{q}^1\xi-V(q^1)$ and $\langle \mathbb{F}l_{\tilde\xi} (q^1,\dot{q}^1,\xi),(q^1,\dot q^1,\eta)\rangle=
\dot{q}^1\eta$. Thus $\F l_{\tilde\xi}$ is not a diffeomorphism.
\end{remark}

It is standard in Routh reduction to introduce the following
definition, although due to the fact that $L$ may not be of type
`kinetic minus potential' the definition is slightly more general
than the usual definition of the {\em locked inertia tensor}.

\begin{definition} Given a Lagrangian $L:TQ \to \mathbb{R}$ and
a point $v_q\in TQ$, the mapping $\mathcal{I}_{v_q}: \mathfrak{g}
\to  \mathfrak{g}^*$ defined by
\begin{eqnarray*}
\mathcal{I}_{v_q}(\xi)(\eta) &=& \left. \frac{d}{d\epsilon}\right|
_{\epsilon =0} \left. \frac{d}{d\tau} \right|_{\tau =0} L(v_q +
\epsilon \xi _Q (q) + \tau \eta _Q (q))\\&=& \left.
\frac{d}{d\tau}\right| _{\tau =0}J_L (v_q + \tau \eta _ Q (q))(
\xi),
\end{eqnarray*}
for $\xi, \eta \in \mathfrak{g}$ arbitrary, is called the
\emph{locked intertia tensor} at $v_q$.
\end{definition}

We say that $L$ is {\em locally}  $G$-regular if the
locked-inertial tensor $\mathcal{I}_{v_q}$ is invertible for any
$v_q \in TQ$. Local $G$-regular Lagrangians only guarantee that
$J_L^{-1}(\mu)$ is locally diffeomorphic to $T(Q/G)\times
Q/G_\mu$. In any case, local $G$-regularity presents the advantage
of being a property easy to check (besides the applications of the
locked inertia tensor in other situations as stability, cf.
\cite{lew}). We refer to~\cite{mestcram} for a description of
Routh reduction for locally $G$-regular Lagrangians.

\subsection{The reduced form induced by the connection}

We now return to the general case and continue defining additional
objects living on the quotient spaces under consideration. These objects
will appear as force terms when studying Routhian reduced systems.
Recall that we fixed a principal connection $\omega$ on $Q$. The one-form on $Q$ obtained by pairing $\mu$ with the $\lag$-valued one-form $\omega$ is denoted by $\omega^\mu=\langle
\mu,\omega\rangle$. The Lie-algebra $\lag _\mu$ of the isotropy subgroup $G_\mu$ is a
subalgebra of $\lag$ which can alternatively be determined as the
set of elements $\xi\in\lag$ satisfying $\mathrm{ad}^*_\xi\mu=0$.
\begin{proposition}
\label{prop7}
 The two-form $d\omega^\mu = \langle \mu,d\omega\rangle$ on
 $Q$ is projectable to a two-form on $Q/G_\mu$, denoted by $\beta^\mu$.\\
 The invariant force form $F:TQ\to T^*Q$ reduces to a map: $f:TQ/G \to T^*(Q/G)$.
\end{proposition}
\begin{proof} It is easily seen that $d\omega^\mu$ is invariant under $G_\mu$ and
annihilates vertical tangent vectors in $Q\to Q/G_\mu$. The latter
follows from the equivariance of $\omega$: let $\xi$ be an
arbitrary element in $\lag_\mu$, then
\begin{eqnarray*}
i_{\xi_Q} d\omega^\mu &=& \langle \mu , {\cal L}_{\xi_Q}\omega
\rangle - \langle \mu , di_{\xi_Q}\omega\rangle\\
&=& \langle \mu, -\mathrm{ad}_\xi \omega\rangle = -  \langle
\mathrm{ad} ^* _\xi \mu, \omega\rangle = 0\ .
\end{eqnarray*}
The second part of the proposition is a straightforward
consequence of Definition \ref{def:invarfor} of an invariant force
term.
\end{proof}
Next we study the bundle $\overline\pi_\mu:Q/G_\mu\to Q/G$ and in
particular the set of tangent vectors vertical to $\overline
\pi_\mu$.

First note that both $\pi:Q\to Q/G$ and $\pi_\mu: Q\to Q/G_\mu$
are principal fibres bundles (from now on PFB), with structure
group $G$ and $G_\mu$ respectively. The adjoint bundle
$\tilde\lag$ is the bundle associated to $\lag$ w.r.t the PFB
$\pi$. Similarly we can consider the adjoint bundle
$\tilde\lag_\mu$ associated to $\lag_\mu$ w.r.t. the PFB
$\pi_\mu$. It is not hard to see that $\tilde\lag_\mu$ is a
subbundle of the fibred product $Q/G_\mu\times\tilde\lag$: let
$[q,\xi]_{G_\mu}\in\tilde\lag_\mu$ be arbitrary, then we define an element in $Q/G_\mu\times \tilde\lag$ by fixing a representative $(q,\xi)\in Q\times \lag_\mu$ and we let $(y=\pi_\mu(q)=y,\tilde\xi=[q,\xi]_G)$ determine the corresponding element in $Q/G_\mu\times\tilde\lag$.  The Lie algebra bracket $[\cdot ,\cdot
]$ on $\lag$ carries over to a Lie algebra structure on the fibres
of $\tilde\lag$. The bracket is denoted with the same symbol:
$[\tilde\xi,\tilde\eta]= [q, [\xi,\eta]]_{G}$, with
$[q,\xi]_G=\tilde\xi$ and $[q,\eta]_G=\tilde\eta$ arbitrary. For
notational convenience we introduce the following shorthand
notations: $\overline\pi_\mu^*\tilde\lag =
Q/G_\mu\times\tilde\lag$, and the dual bundle by
$\overline\pi_\mu^*\tilde\lag^* = Q/G_\mu\times\tilde\lag^*$.
Since $\tilde\lag_\mu$ is a subbundle of
$\overline\pi_\mu^*\tilde\lag$ we may consider the quotient
bundle: $\overline\pi_\mu^*\tilde\lag/\tilde\lag_\mu$. Finally,
note that $\tilde\mu:Q/G_\mu \to \tilde\lag^*$ can be seen as
section of $Q/G_\mu\times\tilde\lag^* \to Q/G_\mu$.
\begin{definition}
Associated with the map $\tilde\mu:Q/G_\mu \to \tilde\lag^*$ we
define a section $\mathrm{ad}^*\tilde\mu$ of
$\bigwedge^2\overline\pi_\mu^*\tilde\lag^*\to Q/G_\mu$ in the
following way
\[
\mathrm{ad}^*\tilde\mu(y)((y,\tilde \xi),(y,\tilde \eta)) =
\langle \tilde \mu(y),[\tilde \xi,\tilde\eta]\rangle,
\]
where $\tilde\xi,\tilde\eta\in \tilde\lag _x, \pi_\mu (y)=x$
arbitrary (recall that $\overline\pi_\mu^*\tilde\lag=Q/G_\mu\times
\tilde \lag$).
\end{definition}
We will use the following shorthand notations:
$\mbox{ad}^*\tilde\mu(y)(\tilde \xi,\tilde \eta)$ or
$\mathrm{ad}^*_{\tilde \xi}\tilde\mu(y)=
i_{\tilde\xi}(\mathrm{ad}^*\tilde\mu(y))$. The two-form
$\mathrm{ad}^*\tilde\mu$ is clearly antisymmetric and its kernel
consists precisely of elements in $\tilde\lag_\mu$. Therefore,
without introducing a new symbol, we can regard
$\mbox{ad}^*\tilde\mu$ as a section of
$\bigwedge^2(\overline\pi_\mu^*\tilde\lag/\tilde \lag_\mu)^*\to
Q/G_\mu$ which is nondegenerate as a map
$\overline\pi_\mu^*\tilde\lag/\tilde \lag_\mu \to
(\overline\pi_\mu^*\tilde\lag/\tilde \lag_\mu)^*$. Recall that by
choosing a point $q\in\pi_\mu^{-1}(y)$, the fibre of the subbundle
$\tilde{\lag}_\mu$ at the point $y$ is isomorphic to $\lag_\mu$,
the Lie algebra of the isotropy subgroup $G_\mu$. This observation
brings us to the next proposition.

\begin{proposition}
\label{repto} The bundle $V\overline\pi_\mu\to Q/G_\mu$ of
vertical tangent vectors to the fibration
$\overline\pi_\mu:Q/G_\mu\to Q/G$ is isomorphic to
$\overline\pi_\mu^*\tilde\lag/\tilde \lag_\mu \to Q/G_\mu$.
\end{proposition}
\begin{proof}
We define a surjective linear map from $\overline\pi^*_\mu\tilde
\lag$ to $V\overline\pi_\mu$, in the following way. Fix a point
$(y,\tilde\xi)\in \overline\pi_\mu^*\tilde\lag$ such that
$y=[q]_{G_\mu}$ and $\overline\pi_\mu(y)=x$. Let $\xi$ be a
representant of $\tilde\xi$ at the point $q$, i.e.
$\tilde\xi=[q,\xi]_G$, and consider the vertical tangent vector in
$V_{y}\overline\pi_\mu$ obtained by $T\pi_\mu(\sigma_q(\xi))$. It
is clear that this map is well defined (i.e. independent of the
choice of the point $q$ in $[q]_{G_\mu}$ and that it is onto.
Furthermore, the kernel of this map consists of the elements in
$\tilde \lag_\mu$ at the point $y$.
\end{proof}
For later purposes it is important to see that the tangent vectors
to the vertical curves
$[qg(t)]_{G_\mu}\in\overline\pi_\mu^{-1}(x)$ through $y$ with
$g(0)=e$ and $\xi=\dot g(0)\in\lag$ are mapped onto $(y,[q,\xi]_G)
+\tilde\lag_\mu$ at the point $y$. We conclude with recalling the
definition of a connection on an arbitrary associated bundle
corresponding to the fixed principal connection $\omega$ on $Q$. Let $E$
denote a manifold on which $G$ acts from the left.

\begin{definition}
The connection or horizontal distribution on the associated bundle
$Q\times_G E$ corresponding with the principal connection $\omega$
is defined as the set of tangent vectors to curves of the form
$[q(t),f]_G$ with $q(t)\in Q$ horizontal and $f\in E$ arbitrary.
In the case that $F$ is a linear space on which $G$ acts linearly,
we define the covariant derivative $D$ as
\[
\frac{D}{Dt} [q(t),e(t)]_G = [q(t),\left(\omega_{q(t)}(\dot
q(t))\right)\cdot e(t)]_G + [q(t),\dot e(t)]_G .
\]
\end{definition}
It is understood that in the above definition the action of a
Lie-algebra element $\omega_q(\dot q)$ on the element $e\in E$ is
denote by $\cdot$.

We will use the connection in associated bundles in the following
two cases:

\noindent Firstly, the principal connection $\omega$ determines a
connection in $Q/G_\mu\to Q/G$, the bundle associated to the orbit
space ${\cal O}_\mu$. The tangent space to $Q/G_\mu$ can be
written as a direct sum of the horizontal and vertical
distribution:
\[
T(Q/G_\mu) \cong \overline\pi_\mu^*T(Q/G)\oplus_{Q/G_\mu}
\overline\pi_\mu^*\tilde\lag/\tilde\lag_\mu.
\]
This in turn implies that the two-form $\beta^\mu$ (cf.
Proposition \ref{prop7} above) can be decomposed into a
horizontal-horizontal and vertical-vertical part (the
horizontal-vertical part vanishes)
\begin{equation}
\label{hhvv} \beta^\mu=(\tilde\Omega^\mu, -\mbox{ad}^*\tilde\mu),
\end{equation}
where
\begin{eqnarray*}
\tilde\Omega^\mu&:&Q/G_\mu \to  \wedge ^2 T^*(Q/G);y  \mapsto  \langle\tilde\mu(y),\tilde\Omega\rangle ,
\end{eqnarray*}
$\tilde\Omega$ being the $\tilde\lag$-valued curvature 2-form on
$Q/G$. Such a decomposition is a straightforward consequence of
the structure equation $\Omega = d\omega + [\omega,\omega]$ and
the fact that in this specific case the horizontal distribution on
$Q/G_\mu$ is the projection of the horizontal distribution on $Q$.

\noindent Secondly, if $E=\lag$ the covariant derivative on the
adjoint bundle $\tilde\lag$ equals, for any $q(t)\in Q$ and
$\xi(t)\in\lag$ (with $[\cdot,\cdot ]$ the Lie bracket on $\lag$):
\[
\frac{D}{Dt} [q(t),\xi(t)]_G = \left[q(t),\left[\omega_{q(t)}(\dot
q(t)), \xi(t)\right]\right]_G + [q(t),\dot \xi(t)]_G.
\]

\subsection{Variations}

\label{gil} The notions introduced here are described in more
detail in~\cite{CMR01}. For the sake of completeness and because
we work in a slightly different situation we nevertheless briefly
introduce the notions of vertical and horizontal variations in the
context of Lagrangian systems with symmetry. The main point is to
consider a variation of a curve $q(t)$ in the total space $Q$, and
subsequently study the projection of this variation to the
quotient spaces $Q/G_\mu$ and $Q/G$.

For that purpose, we fix a curve in $q(t)$ in $Q$ and its lift
$\dot q(t)$ in $TQ$. Following the notations from the previous
section, we shall write the projection of $\dot q(t)$ to
$TQ/G_\mu$ as $(\dot x(t), y(t),\tilde \xi(t))$, where $\dot
x(t)=T_{q(t)}\pi(\dot q(t))$, $y(t)=\pi_\mu(q(t))$ and $\tilde
\xi(t)=[q(t),\omega_{q(t)}(\dot q(t))]_G$.

\begin{definition}
A deformation of the curve $q(t)$ is smooth function
$q(t,\e)=q_\e(t)$ such that $q(t,0)=q(t)$. The corresponding
variation is defined by
\[
\delta q(t)= \left.\fpd{q(t,\e)}{\e}\right|_{\e=0}\in T_{q(t)}Q\ .
\]
\begin{enumerate}
\item A vertical deformation is a deformation that can be written
as $q(t,\epsilon )= q(t)g_\epsilon (t)$, where the family of
curves $g_\epsilon (t) \in G$ satisfies $g_0(t)=e$, with $e$ the
unity of $G$.
\item A variation $\delta q(t)$ such that
$T\pi_{q(t)} (\delta q(t))=0$ for all $t$ is said to be a vertical
variation. \item A variation $\delta q(t)$ such that
$\omega(q(t))(\delta q(t))=0$ for all $t$ is said to be a
horizontal variation.
\end{enumerate}
\end{definition}

Obviously, the variation induced by a vertical deformation is
vertical and it is not hard to see that any vertical variation can
be given in that way.

Given an arbitrary vertical deformation $q(t,\e)=q(t)g_\e(t)$,
then the projection of $(q(t,\e),\dot q(t,\e))$ to the quotient
space $TQ/G_\mu$ is written as $(\dot x(t), y(t,\e),\tilde
\xi(t,\e))$, where according to the previous definitions:
\begin{eqnarray*}
  y(t,\e)& = &[q(t,\e)]_{G_\mu}\\
  \tilde \xi(t,\e) &=& [q(t)g_\e(t),\omega(q(t,\e))
  (\dot  q(t,\e))]_G \\ &=& \left[q(t)g_\e(t),\omega(q(t)g_\e(t))\left(TR_{g_\e(t)}
  \left(\dot q(t) + \sigma_{q}(\dot g_\e g_\e^{-1})\right)\right)\right]_G\\
  &=& \left[q(t),\omega(q(t))\big(\dot q(t)\big) + \dot g_\e(t)g_\e^{-1}(t)\right]_G
  =\tilde\xi(t) + [q(t),\dot g_\e(t)g_\e^{-1}(t)]_G
\end{eqnarray*}
The typical structure equations of {\bf reduced vertical
variations} is obtained in the following way. Let $d/d\e _{\e =0}
g_\e(t)= \delta g(t)\in\lag$ and $\tilde\eta(t)=[q(t),\delta
g(t)]_G\in\tilde\lag$. Using the definition of the covariant
derivative on the associated bundle $\tilde\lag\to Q/G$, we may
write
\[
\frac{D}{Dt} \tilde\eta(t)= \left[q(t),\left[\omega(q(t))\big(\dot
q(t)\big),\delta g(t)\right]\right]_G+ [q(t),\dot{\delta g}(t)]_G.
\]
Finally, by computing the derivative to the curve $\e\to
\tilde\xi(t,\e)$ we find the structure equation~\cite{CMR01}:
\[
\tilde\lag_{x(t)}\ni \delta\tilde\xi(t):=
\left.\frac{d}{d\e}\right|_0 \tilde \xi(t,\e) = [q(t),\dot \delta
g(t)]_G = \frac{D}{Dt} \tilde \eta(t) - [\tilde \xi (t),\tilde
\eta(t)].
\]

It now remains to check the variation in the variable $y$. The
curve $\e\mapsto y(t,\e) = [q(t,\e)]_{G_\mu}$ is contained in a
fibre of the bundle $\overline \pi_\mu^{-1}(x(t))$ and
consequently, the tangent vector to this curve is in
$V\overline\pi_\mu$.  In the language of Proposition \ref{repto},
the variation $\delta y$ in $y$ then corresponds to
$(y,\tilde\eta)+\tilde\lag_\mu$.

Next, we describe the structure of reduced {\bf horizontal
variations}. Let $q(t,\e)$ denote a deformation, with
corresponding horizontal variation $\delta q(t)$. Again, the
projection of the deformation $q(t,\e)$ determines a deformation
$(\dot x(t,\e), y(t,\e),\tilde \xi(t,\e))$ of $(\dot
x(t),y(t),\tilde\xi(t))$ in $TQ/G_\mu$ (with
$x(t,\e)=\pi(q(t,\e))$). By definition we have that $\tilde \xi
(t,\e)= [q(t,\e),\omega(q(t,\e))(\dot q(t,\e))]_G$. Using the fact
that $\delta q$ is horizontal, i.e. $\delta q=\delta x^h_q$, the
equality
\[
  \left.\frac{d}{d\e}\right|_0 \omega(q(t,\e))(\dot q(t,\e)) =
  \Omega(q(t))(\delta q(t),\dot q(t))
\]
holds and we are able to write the horizontal and vertical parts
of the tangent vector $\delta \tilde\xi (t) \in
T_{\tilde\xi(t)}\tilde\lag$ as follows:
\begin{eqnarray*}
\delta \tilde \xi(t)^v=\left(\left.\frac{d}{d\e}\right|_{\e
=0}\tilde \xi (t,\e)\right)^v &=& \left(\tilde\xi(t),[q(t), \Omega(q(t))(\delta q(t),\dot q(t))]_G\right)\\
& =&\left(\tilde\xi(t),\tilde \Omega(x(t))( \delta x(t),\dot x(t))\right),\\
\delta \tilde \xi(t)^h=\left(\left.\frac{d}{d\e}\right|_{\e
=0}\tilde \xi(t,\e)\right)^h &=& (\tilde\xi,\delta x(t)),
\end{eqnarray*}
where we used the decomposition of $T_{\tilde\xi(t)}\tilde\lag =
\tilde\pi^*T(Q/G)\oplus V\tilde\pi$ into its horizontal and
vertical subspace (due to the linear structure, $V\tilde\pi$ is
identical to $\tilde\lag\times\tilde\lag$). It is important to see
that $\delta\tilde\xi$ is not necessarily horizontal w.r.t the
connection on the bundle $\tilde\lag\to Q/G$, despite the fact
that it originates from a horizontal variation of $q(t)$.

The deformation $y(t,\e)$ has a corresponding variation $\delta
y(t)$ which is a tangent vector along $y(t)$. It is horizontal in
terms of the connection associated to $\omega$ on the associated
bundle $Q/G_\mu \to Q/G$.

\subsection{Lagrange-Poincar\'e reduction}\label{sec:lagpoin}

From the above structural equations for projected vertical and
horizontal variations of a curve $q(t)$, we are now able to deduce
the Lagrange-Poincar\'e equations on $TQ/G\cong
T(Q/G)\times\tilde\lag$ for an invariant Lagrangian system
$(Q,L,F)$ on $Q$. This follows easily from the following identity:
given an arbitrary deformation of $q(t)$, then
\[
\int_I L(\dot q(t,\e))dt = \int _I l(\dot
x(t,\e),\tilde\xi(t,\e))dt .
\]
Now, if $q(t)$ is critical for $L$, then for any vertical
variation with vanishing endpoints we write, with a slight abuse
of notations
\begin{eqnarray*}
0&=&  \frac{d}{d\e} \int_I L(\dot q(t,\e))= \frac{d}{d\e} \int_I l
(\dot x(t),\tilde\xi(t,\e))dt \\ &=& \int_I  \left\langle
\F_{\tilde\xi} l(\dot x(t),\tilde \xi(t)), \frac{D}{Dt} \tilde
\eta(t) - [\tilde \xi (t),
\tilde \eta(t)]\right\rangle dt\\
&=&\int_I \left\langle -\frac{D}{Dt} \F_{\tilde\xi} l -
\mathrm{ad}^*_{\tilde \xi } \F_{\tilde\xi} l,\tilde
\eta(t)\right\rangle dt,
\end{eqnarray*}
for all $\tilde \eta(t)\in\tilde\lag$ with vanishing endpoints. On
the other hand if we consider horizontal variations, then
\begin{eqnarray}\nonumber
0&=&  \frac{d}{d\e} \int_I l(\dot x(t,\e),\tilde\xi(t,\e))dt
+\int_I \langle f(\dot x(t),\tilde\xi(t)),\delta x(t)\rangle dt \\
\label{eqn:lp} &=& \int_I \left(\left\langle {\cal EL}(l)^h +f
,\delta x \right\rangle +\left\langle  \F_{\tilde\xi} l(\dot
x,\tilde \xi), \tilde \Omega(x)( \delta x,\dot x))\right\rangle
\right)dt.
\end{eqnarray}
Recall that $l$ can be regarded as an intrinsically constrained Lagrangian on
$M=\tilde\lag$ (and $N=Q/G$), and hence ${\cal EL}(l)$ in the previous equation is the
Euler-Lagrange operator for $l$ regarded as a function on
$T\tilde\lag$. The reduced equations now read
\begin{eqnarray*}
\frac{D}{Dt}\mathbb{F}_{\tilde{\xi}}l & = & - \mathrm{ad}^* _{\tilde{\xi}} \mathbb{F}_{\tilde{\xi}}l, \mbox{ and}\\
\mathcal{EL}(l)^h &=& \langle \mathbb{F}_{\tilde{\xi}}l,i_{\dot x}
\tilde{\Omega}\rangle -f.
\end{eqnarray*}
It should be clear that these equations are not the Euler-Lagrange
equations for the intrinsically constrained Lagrangian system
$(\tilde \lag \to Q/G,l,f)$ on the fibred manifold $\tilde \lag
\to Q/G$. The introduction of the Routhian will avoid this
obstacle, paying the price that additional non-conservative forces
have to be taken into account.

\section{The Routhian reduction scheme}

Consider a critical curve $q(t), t \in[a,b]$, of the invariant
Lagrangian system $(Q,L,F)$, that is
\[
0= \delta\int_I L(\dot q(t))dt = -\int_I\langle F(\dot
q(t)),\delta q(t)\rangle dt,
\]
with $\delta q(t)$ an arbitrary variation of $q(t)$. By
considering variations with fixed endpoints, we obtain the Euler-Lagrange equations for critical curves
\[
{\cal EL}(L)(\ddot q(t))= -F(\dot q(t)), \ \forall t\in I.
\]

We assume throughout this section that we fixed a principal
connection $\omega$ and a regular momentum value $\mu \in\lag^*$
of $J_L$.

\begin{definition}
The Routhian is the function on $TQ$ defined by
\[
R^\mu(v_q) = L(v_q)-\langle \mu,\omega(q)(v_q)\rangle.
\]
\end{definition}
The Routhian depends on the choice of the connection and the
momentum $\mu$. Although the Routhian is not in general invariant
under the action of $G$, we do have the property that the momentum
map of the Routhian w.r.t the action of $G$ will be conserved (cf.
proposition~\ref{prop:cm}) . We thus study the behavior of the
Routhian under the action of $G$.
\begin{proposition}
The Routhian is invariant under the action of $G_\mu$. It
transforms under the infinitesimal action of $G$ as:
\[
\xi_{TQ}\left( R^\mu\right)(v_q) = \langle \mathrm{ad}_\xi^*
\mu,\omega(q)(v_q)\rangle , \quad \forall \xi \in \lag .
\]
\end{proposition}
\begin{proof}
This is straightforward from the definition of $R^\mu$ and the
equivariance of $\omega$.
\end{proof}

\begin{proposition}\label{prop:equiv}\hfill
 \begin{enumerate}
\item A critical curve of the invariant Lagrangian system $(Q,L,F)$ is a
critical curve of the Lagrangian system $(Q,R^\mu,F+G^\mu)$ and
vice versa, where $G^\mu$ is the gyroscopic force term (see \S
\ref{sec:velcon} above) associated to $d\omega^\mu$, that is
\[
G^\mu (v_q ) = -i _{v_q}d\omega ^\mu  = -\langle \mu,i_{v_q} d\omega \rangle , \quad v_q \in T_qQ.
\]

\item The Lagrangian system $(Q,R^\mu,F+G^\mu)$ has conserved
momentum $J_{R^\mu}$ w.r.t. the action of the entire symmetry
group $G$.
\item The momentum maps of $L$ and $R^\mu$, both considered with respect to
the full action of $G$, are related by $J_{R^\mu}=J_L-\mu$. Hence $q(t)$ is critical curve
of the invariant Lagrangian system $(Q,L,F)$ with momentum $J_L=\mu$
if and only $q(t)$ is a critical curve of
Lagrangian system $(Q,R^\mu,F+G^\mu)$ with momentum $J_{R^\mu}=0$.
 \end{enumerate}
\end{proposition}
\begin{proof}
The first statement is proven by considering the variation of
$R^\mu$ along a critical curve $q(t)$ of $(Q,L,F)$: given an
arbitrary variation $\delta q(t)$  of $q(t)$ in $Q$ then
\begin{eqnarray}\label{variations}\nonumber
\delta \int_I R^\mu(\dot q(t)) dt &=&  \delta \int_I L(\dot q(t))
dt
-\delta \int_I \langle \omega^\mu(q(t)), \dot q(t)\rangle dt \\
\nonumber & =& \delta \int_I L(\dot q(t)) dt + \int_I \langle
i_{\dot q(t)} d\omega^\mu(q(t)),
\delta q(t)\rangle dt -\langle \omega^\mu(q(t)),\delta q(t)\rangle^b_a\\
&=&  -\int_I\langle F(\dot q(t)),\delta q(t)\rangle dt-  \int_I
\langle G^\mu(\dot q(t)), \delta q(t)\rangle dt \nonumber\\ & &
\qquad \qquad -\langle \omega^\mu(q(t)),\delta
q\rangle^b_a +\langle \F L(\dot q(t)),\delta q(t)\rangle|^b_a\ .
\end{eqnarray}
This easily proves the first statement in the proposition (the
other direction follows by reversing the arguments).

For the second statement, according to Proposition \ref{prop:cm},
it is sufficient to show that $\xi_{TQ}(R^\mu)(v_q) = -\langle
F+G^\mu,\xi_Q\rangle(q)$. The contraction of the force term
$F+G^\mu$ with vertical directions is precisely
\begin{eqnarray*}
\langle (F+G^\mu)(v_q),\xi_Q(q)\rangle &=& \langle G^\mu(v_q),\xi_Q(q)\rangle\\
&=& \langle \mu, [\omega(q)(v_q),\xi]\rangle\\ &=& -\langle
\mathrm{ad}^*_\xi\mu,\omega(q)(v_q)\rangle.
\end{eqnarray*}
Together with the previous proposition this shows the second
statement.

Finally, the third statement is proven by computing the momentum
map w.r.t action of $G$ of the Lagrangian system
$(Q,R^\mu,F+G^\mu)$:
\begin{eqnarray*}
\langle
J_{R^\mu}(v_q),\xi\rangle&=&\left.\frac{d}{d\e}\right|_{\e=0}
\bigg(L(v_q+\e\sigma_q(\xi)) - \left\langle
\mu,\omega(q)\big(v_q+\e\sigma_q(\xi)\big)\right\rangle \bigg)\\&=&
\langle J_L(v_q) -\mu,\xi\rangle.\end{eqnarray*}
\end{proof}

We are now ready to state the first theorem of this paper. First
we need to fix additional notations. The Routhian is reducible to
a function
\[
\Ro^\mu : TQ/G_\mu\cong T(Q/G)\times Q/G_\mu\times \tilde \lag \to
\R.
\]
Recall that the invariant force term $F$ reduces to a map
\[
f:TQ/G\to T^*(Q/G)
\]
and that the gyroscopic force term $G^\mu$ reduces to a gyroscopic
force term
\[
\zeta^\mu:T(Q/G_\mu) \to T^*(Q/G_\mu)
\]
on $Q/G_\mu$ associated to $\beta^\mu$ (see Proposition
\ref{prop7}). With these notations and those in \S
\ref{sec:lagsyst} in mind, we now consider the bundle $M\to Q/G$,
where
\[
M:=Q/G_\mu\times\tilde\lag \to Q/G.
\]
The triple  $(M\to Q/G,\Ro^\mu,f+\zeta^\mu)$ will determine an
intrinsically constrained Lagrangian system with
\[
T_M N= T(Q/G)\times Q/G_\mu \times \tilde{\mathfrak{g}}.
\]
To reduce the notational complexity, it is understood that $f$ and
$\zeta^\mu$ are pull-backed to the appropriate bundles in order to
fit the definition of an intrinsically constrained Lagrangian
system.

Due to the product structure of the bundle
$M=Q/G_\mu\times\tilde\lag$, we can restrict the Euler-Lagrange
operator ${\cal EL}(\Ro^\mu)$ to tangent vectors vertical to $M\to
 \tilde\lag$ or to $M\to Q/G_\mu$. The latter is denoted by $\partial_y
\Ro^\mu$ and the first is precisely $\F_{\tilde\xi} \Ro^\mu$. The
two bundles $Q/G_\mu$ and $\tilde\lag$ are equipped with a
connection determined by $\omega$, and it is standard to show that
these determine a unique connection on $M$ for which the
horizontal distribution projects onto the horizontal distribution
of $Q/G_\mu$ and $\tilde\lag$. It should be clear what is meant
when we write ${\cal EL}(\Ro^\mu)^h$ and ${\cal EL}(\Ro^\mu)^v =
(\partial_y\Ro^\mu, \F_{\tilde\xi}\Ro^\mu)$.

\begin{theorem}\label{thm:routh}
(i) A critical curve of the invariant Lagrangian system $(Q,L,F)$
with momentum $\mu$ projects onto a critical curve of the
intrinsically constrained Lagrangian system $(M\to
Q/G,\Ro^\mu,f+\zeta^\mu)$. (ii) Conversely, if a critical curve of
the intrinsically constrained Lagrangian system $(M\to
Q/G,\Ro^\mu,f+\zeta^\mu)$ is the projection of a lifted curve
$\dot q(t)$ in $TQ$, then $q(t)$ is a critical curve of $(Q, L,F)$
with momentum $\mu$.
\end{theorem}
\begin{proof}
The proof of this theorem is obtained by first considering the
Lagrange-Poincar\'e reduction of the Lagrangian system $(Q,R^\mu,
F+G^\mu)$ w.r.t the action of $G_\mu$, and consequently  by showing
that the reduced equations for critical curves with zero momentum
coincide with the critical curves for the intrinsically
constrained Lagrangian system $(M\to Q/G, \Ro^\mu,f+\zeta^\mu)$.

Let $q(t): I \to Q$ denote a critical curve of $(Q,L,F)$ with
momentum $\mu$. From proposition~\ref{prop:equiv}, we know that
$q(t)$ is a critical curve of the Lagrangian system $(Q,R^\mu,
F+G^\mu)$ with momentum $0$.

The projection of $\dot q(t)$ onto $M$ is denoted by $(y(t),\tilde
\xi(t))$. Recall that the vertical part of $\dot y(t)$ equals
$(y(t),\tilde \xi(t)) +\tilde \lag_\mu$ (this is the reduced
version of the fact that the curve in $TQ$ is the complete lift of
a curve in $Q$, see infra).  By definition of $\Ro^\mu$, we may
write $R^\mu(\dot q(t)) = \Ro^\mu(\dot x(t),y(t),\tilde \xi(t))$,
with $x(t)=\pi(q(t))$. We now show (i) by studying the reduction
of the variational equation for critical curves of
$(Q,R^\mu,F+G^\mu)$ with zero momentum (similar to
section~\S\ref{sec:lagpoin} for the Lagrangian system $(Q,L,F)$). Consider an arbitrary variation $\delta q(t)$ of $q(t)$ and the corresponding reduced variation $(\delta y(t), \delta \tilde\xi(t))$ of $(y(t),\tilde\xi(t))$, then (for notational convenience we omit the
explicit time dependence)
\begin{eqnarray}\nonumber
& \delta& \int_I R^\mu(\dot q) dt + \int_I \langle (F+G^\mu)(\dot
q),\delta q\rangle dt \\ \label{eq:varal} \quad &=&  \delta \int_I\Ro^\mu(\dot x,y,\tilde \xi)
dt  + \int_I \left( \langle f(\dot x,\tilde\xi),\delta
x\rangle+\langle\zeta^\mu(\dot y),\delta y\rangle \right)dt.
\end{eqnarray}

\noindent \emph{Vertical Variations:} We apply the Lagrange-Poincar\'e reduction scheme and we assume $\delta q(t)$ to be a vertical variation. From the previous
section, we know that the reduced variations $\delta
y(t)$ of $y(t)$ and $\delta \tilde \xi(t)$ of $\tilde \xi(t)$ both
satisfy the following equalities:
\begin{eqnarray*}
&&\delta\tilde\xi(t)= \frac{D}{Dt} \tilde \eta(t) - [\tilde \xi (t),\tilde \eta(t)]\\
&&\delta y(t) = (y,\tilde \eta(t)) + \tilde \lag_\mu\\
&&\delta x(t) =0
\end{eqnarray*}
where $\tilde\eta(t)$ can be chosen arbitrarily (with or without
vanishing endpoints). From Eq.~(\ref{eq:varal}), we conclude that
\[
\delta \int_I\Ro^\mu(\dot x,y,\tilde \xi) dt  + \int_I\left(
\langle f(\dot x,\tilde\xi),\delta x\rangle+\langle\zeta^\mu(\dot
y),\delta y\rangle \right) dt=0
\]
for {\em reduced} vertical variations of $(y(t),\tilde\xi(t))$.
Moreover, since the variations are vertical w.r.t the fibration $M\to Q/G$, the first term is
\[
\delta \int_I\Ro^\mu(\dot x,y,\tilde \xi) dt  = \int_I \langle
\partial_y\Ro^\mu, \delta y\rangle + \langle \F_{\tilde\xi}
\Ro^\mu,\delta\tilde\xi\rangle dt.
\]
From the fact that $\F _{\tilde\xi} \Ro^\mu = \F_{\tilde\xi} l -
\tilde\mu=0$, the variational equation holds for {\em arbitrary}
vertical variations $(\delta y,\delta\tilde\xi)$ of
$(y,\tilde\xi)$, and not only those satisfying $\delta\tilde\xi =
\frac{D}{Dt} \tilde \eta(t) - [\tilde \xi (t),\tilde \eta(t)]$.
Therefore, the variational equation of the intrinsically
constrained system $(M\to Q/G,\Ro^\mu,f+\zeta^\mu)$ is satisfied
for arbitrary vertical variations in $(y(t),\tilde\xi(t))$ (i.e. vertical w.r.t the bundle $M\to Q/G$).

\noindent \emph{Horizontal variations.} An arbitrary horizontal
variation of $q(t)$ determines, after reduction to $M$, an
arbitrary horizontal variation of $y(t)$ in $Q/G_\mu$. The
corresponding variation of $\tilde\xi(t)$ is not horizontal w.r.t the associated connection on $\tilde\lag$ but
the contribution of the vertical part to the variational equation
will vanish since $\F_{\tilde \xi}\Ro^\mu =0$. We then have ${\cal
EL}(\Ro^\mu)^h = -((\zeta ^\mu )^h +f)$ which is
\[
{\cal EL}(\Ro^\mu)^h = i_{\dot{x}} \tilde{\Omega} ^\mu (y) -f(\dot
x,\tilde\xi)
\]
due to formula~(\ref{hhvv}).

(ii) The converse statement is easily shown in the following way.
Let $\dot q(t)$ denote a curve in $Q$ that projects onto $(\dot
x(t),y(t),\tilde \xi(t))$, the latter curve being a a critical
curve of the intrinsically constrained system on
$Q/G_\mu\times\tilde\lag$. This implies that
\[
\delta \int_I\Ro^\mu(\dot x,y,\tilde \xi) dt  + \int_I\langle
f(\dot x,\tilde\xi),\delta x\rangle+\langle\zeta^\mu(\dot
y),\delta y\rangle dt  =0
\]
for an {\em arbitrary} variation of the critical curve. If we
consider vertical variations to $(y(t),\tilde\xi(t))$ for which
$\delta y(t)=0$, then $\F_{\tilde\xi} \Ro^\mu=j_l(\dot
x,\tilde\xi)-\tilde\mu(y)=0$. In terms of the curve upstairs, this
is precisely $J_L(\dot q(t))=\mu$. Assume that we restrict
the class of variations of $(y(t),\tilde\xi(t))$ to the projected
 variations of $q(t)$ living in the total space $Q$. By using
the invariance of $R^\mu$, it is now easily seen from the previous
variational equation for $\Ro^\mu$ that
\[
\delta \int_I R^\mu(\dot q) dt + \int_I \langle (F +G^\mu)(\dot
q), \delta q\rangle dt=0.
\]
Proposition~\ref{prop:equiv} concludes the proof.
\end{proof}

In the previous Theorem we introduced the reduction of a Lagrangian system $(Q,L,F)$
towards an intrinsically constrained Lagrangian system on a bundle $M$
over $Q/G$. It should be clear that due to the product structure
of $M$, there are two intrinsic constraints: $\F_{\tilde\xi}
\Ro^\mu =0$ or $j_l(\dot x,\tilde\xi)=\tilde\mu(y)$ and
$\partial_y\Ro^\mu + \zeta^\mu(\dot y)^v =0$. The latter constraint
belongs to the gyroscopic class (see section~\ref{sec:velcon}) and in the following Proposition we study this constraint in more detail.
\begin{proposition}
Let $(v_x,y,\tilde\xi)$ be arbitrary in $T_M(Q/G)=T(Q/G)\times Q/G_{\mu}\times
\tilde{\mathfrak{g}}$ and
$(y,\tilde\eta)$ in $M=Q/G_\mu \times \tilde{\mathfrak{g}}$
with $\tilde\eta$ arbitrary, then
\[
\langle \partial_y \Ro^\mu(v_x,y,\tilde\xi),
(y,\tilde\eta)+\tilde\lag_\mu\rangle = -\langle
\mathrm{ad}^*_{\tilde \xi}\tilde \mu(y),\tilde \eta\rangle.
\]
\end{proposition}
\begin{proof}
The proposition follows from the following computations, with
$\tilde\eta=[q,\dot g(0)] \in\tilde\lag$ arbitrary for some curve
$g(\e)$ in $G$ through the identity at $\e=0$:
\begin{eqnarray*}
\langle \partial_y \Ro^\mu,(y,\tilde\eta)+\tilde\lag_\mu\rangle &=&
-\langle \partial_y\langle \tilde
\mu(y),\tilde\xi\rangle,(y,\tilde\eta)+\tilde\lag_\mu\rangle \\
&=&\left. \frac{d}{d\e}\right|_{\e =0} [q,\langle \mu, Ad_{g(\e)}
\xi\rangle]= \langle \mathrm{ad}^*_{\tilde \eta}\tilde\mu,\tilde
\xi\rangle= -\langle \mathrm{ad}^*_{\tilde \xi}\tilde \mu,\tilde
\eta\rangle.
\end{eqnarray*}
\end{proof}

To provide a complete insight into the intrinsic constraint, we now
compute $(\zeta^\mu)^v$. Let $v_y \in T_y(Q/G_\mu)$ be arbitrary,
then from formula~(\ref{hhvv}) we find
\[
\langle \zeta ^\mu(v_y),(y,\tilde\eta)+\tilde\lag_\mu\rangle = -
\beta^\mu\big(v_y, (y,\tilde\eta)+\tilde\lag_\mu\big)= \langle
\mathrm{ad}^*_{v_y^v} \tilde{\mu}, \tilde{\eta}\rangle.
\]
The constraint on critical curves $(y(t),\tilde\xi(t))$ then
reduces to $\mathrm{ad}^*_{\dot y^v(t)}
\tilde{\mu}(y(t))=\mathrm{ad}^*_{\tilde{\xi}(t)}
\tilde{\mu}(y(t))$. Due to the nondegeneracy of $\mathrm{ad}^*$,
this is equivalent to saying that $\dot y^v(t) =
(y(t),\tilde\xi(t))+\tilde\lag_\mu$.

We have thus proven the following reduction technique. Because it preserves the Lagrangian nature we refer to it as Routh reduction.

\begin{theorem}\label{thm:routhsplit}
Any critical curve $q(t)$ of an invariant
Lagrangian system $(Q,L,F)$ with momentum $\mu$ projects onto a critical curve
$(y(t),\tilde\xi(t))$ of the intrinsically constrained Lagrangian
system $(M\to Q/G,\Ro^\mu, f+\zeta^\mu)$. The Euler-Lagrange
equations for $(M\to Q/G,\Ro^\mu, f+\zeta^\mu)$ are
\begin{eqnarray}
 &&{\cal EL}(\Ro^\mu)^h\left(\ddot x,\dot
y,\dot{\tilde\xi}\right)=
i_{\dot x}\tilde \Omega^\mu(y)-f(\dot x,\tilde\xi) \nonumber\\
&&j_l(\dot x,\tilde\xi)=\tilde \mu(y)\nonumber\\
&&\dot y^v = (y,\tilde \xi) + \tilde\lag_\mu.
\end{eqnarray}
\end{theorem}

\subsection{Remarks on Routh reduction}

\begin{remark}
\textnormal{First we would like to draw the attention to the
structure of the Routhian $R^\mu=L-\omega^\mu$. If we start with a
critical curve of the Lagrangian system $(Q,L , F)$ with momentum
$\mu ' \neq \mu$ (that is, a critical curve of $(R^\mu, F-G^\mu)$
with non zero momentum), the projected curve
$(y(t),\tilde{\xi}(t))$ is a critical curve of a variational
problem for the system $(M\to Q/G,\mathcal{R}^\mu , f + \zeta
^\mu)$ with constraints in the variations for $\tilde{\xi}$. The
constraints are those described in Section~\S\ref{sec:lagpoin} but
w.r.t the action of $G_\mu$. Hence, $(y(t),\tilde{\xi}(t))$ may
not be a critical solution of the free problem defined by $(M\to
Q/G,\mathcal{R}^\mu , f + \zeta ^\mu)$. The essential point is
that, for momentum $\mu ' = \mu$, we have $\F_{\tilde \xi} \Ro^\mu
=0$ along such projected critical curves; the constraint on
$\delta \tilde\xi$ becomes negligible and the statement holds for
{\em arbitrary} variations of $(y(t),\tilde\xi(t))$ in $M$.
Moreover, the fact that $\F_{\tilde \xi} \Ro^\mu =0$ allows us to
separate the variations in the variables $y$ and $\tilde \xi$.
This is not possible if one works with the reduced lagrangian $l$.
We could summarize this by saying that the Routhian $R^\mu$ is an
alternative Lagrangian (up to gyroscopic forces) for $L$ such that
critical curves with momentum $\mu$ become critical curves with
zero momentum. Precisely this zero momentum condition will ensure
that there are no constraints in the reduction of the variations.}
\end{remark}

\begin{remark}
\textnormal{Next, we wish to explicitly mention the case where $G$
is Abelian (for instance the case with cyclic coordinates). This
assumption simplifies matters significantly. Firstly, we have that
$G_\mu=G$, $\tilde\lag =\tilde\lag_\mu$. This implies that the map
$\overline\pi_\mu$ is the identity and that the constraint $\dot
y^v= (y,\tilde\xi)+\tilde\lag_\mu$ in the Routhian reduced system
$(M\to Q/G,\Ro^\mu,f+\zeta^\mu)$ is trivial. The equations of
motion for this system then become:
\begin{eqnarray*}
&&{\cal EL}(\Ro^\mu)^h=i_{\dot x}\tilde \Omega^\mu(x)-f\\
&&j_l(\dot x,\tilde\xi)=\tilde\mu .
\end{eqnarray*}}\end{remark}

\begin{remark}\label{rem:leftright}
\textnormal{In order to relate the results presented here with the
existing literature, we give an outline of what happens when the
action of $G$ on $Q$ is a left action. The structure of the
Routhian reduced equations remains the same, except that some
objects will change sign due to the fact that the structure
equation now takes the form $d\omega=\Omega+[\omega,\omega]$. This
will affect the reduced objects in play, for instance the 2-form
$\beta^\mu$ becomes
\[
\beta^\mu = (\tilde\Omega^\mu, \mathrm{ad}^*\tilde\mu).
\]
The Routhian reduced equations of motion are
\begin{eqnarray*}
  &&{\cal EL}(\Ro^\mu)^h=i_{\dot x}\tilde \Omega^\mu(x)-f\\
  &&j_l(\dot x,\tilde\xi)=\tilde \mu(y)\\
  && \partial_y \Ro^\mu = \mathrm{ad}^*_{\dot y^v}\tilde\mu(y).
\end{eqnarray*}
Note that also $\partial_y\Ro^\mu$ also changes sign in comparison
with the right invariant case, so the third equation is equivalent
to $\dot y^v = (y,\tilde \xi) + \tilde\lag_\mu$.}\end{remark}
\begin{figure}[htb]\centering
\includegraphics{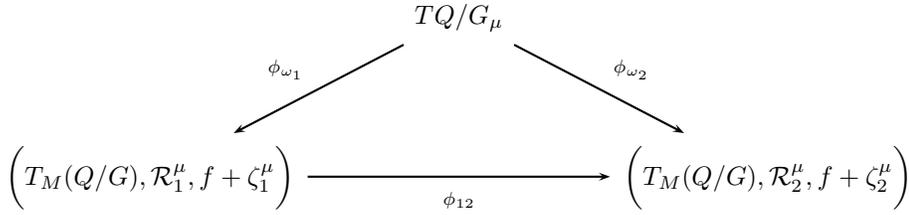}
\caption{Different realizations of $TQ/G_\mu$.}\label{fig:schema2}
\end{figure}
\begin{remark}\label{rem:altrouth}
\textnormal{The choice of the connection $\omega$ will alter the
Routhian and force term in the reduced Lagrangian system in such a
way that the critical curves of the intrinsically constrained
Lagrangian systems coincide. Assume that two principal connections
$\omega_1$ and $\omega_2$ are chosen. Then from standard
connection theory both connections are equal up to an equivariant
$\lag$-valued one-form $\delta$, i.e. $\omega_2=\omega_1 +
\delta$, with $R^*_g\delta = Ad_{g^{-1}} \cdot \delta$ and
$\langle\delta,\sigma(\xi)\rangle =0$ for all $\xi\in\lag$. The
one-form $\delta$ reduces to a $\tilde\lag$-valued one-form on
$Q/G$, denoted by $\tilde\delta$. Consider the one-form
$\tilde\delta^\mu=\langle \tilde \mu ,\tilde\delta\rangle$ on
$Q/G_\mu$ (in fact it is a map from $Q/G_\mu$ to $T^*(Q/G)$). We
now study the effect on the reduction process. First, the connections give
two possible identifications $TQ/G_\mu \simeq T(Q/G)\times Q/G_\mu \times
\tilde{\mathfrak{g}}$. The map $\phi_{12}$ (see Fig.~\ref{fig:schema2})
assumes the following form
\[
\phi_{12}\big(v_x,y,\tilde\xi\big) =
\big(v_x,y,\tilde\xi+\tilde\delta(v_x)\big).
\]  If we consider the
Routhian functions $\Ro^\mu_1$, $\Ro^\mu_2$ and the force terms
$\zeta^\mu_1$, $\zeta^\mu_2$ associated to $\omega_1$ and
$\omega_2$ respectively, then
\begin{eqnarray*}
&&\big(\phi^*_{12}\Ro^\mu_2\big)(v_x,y,\tilde\xi) =
\Ro^\mu_1(v_x,y,\tilde\xi)
-\langle \tilde \delta^\mu(y),v_x\rangle\\
&&\zeta^\mu_2(v_y) = \zeta^\mu_1(v_y) + i_{v_y}d\tilde\delta^\mu.
\end{eqnarray*}
This shows that both (intrinsically constrained) Lagrangian
systems are equivalent if we remember the fact that a linearly
velocity dependent potential in a Lagrangian is equivalent to the
addition of a gyroscopic force term (see~\cite{pars}). The choice
of the connection can be used to simplify the Routhian (paying the
price that a gyroscopic force term is added to the system) or to
simplify the gyroscopic force terms (paying the price that the
Routhian will contain linear terms in the velocity). We return to
this fact when we consider classical Routhian reduction and
compare some recent geometric results in the literature and the
formulation of the result as stated in~\cite{pars}.}
\end{remark}

\begin{remark}
\textnormal{We end this section with a presymplectic formulation
of the Routh reduced system. Using the notations from
Section~\ref{sec:velcon} we can consider the pull-back of the
canonical symplectic form $\omega_{Q/G}$ from $T^*(Q/G)$ to
$T_M(Q/G)$ using the Legendre transform of $\Ro^\mu$. Since the
intrinsic constraint associated to the fibration $Q/G_\mu\to Q/G$
comes from a regular gyroscopic term, the following theorem is a
direct consequence of the definitions in Section~\ref{sec:velcon}.
Here we use the projections $\overline\pi_1: T_M^*(Q/G)\to
T^*(Q/G)$, $\pi_2: M\to Q/G_\mu$.
\begin{theorem}\label{thm:presymp}
The Routhian reduced system of a Lagrangian system $(Q,L,F)$ can
be formulated into a presymplectic system on $TQ/G_\mu\cong
T_M(Q/G)$, where the presymplectic two-form is
$\omega_{\Ro^\mu}=(\overline\pi_1\circ\F_1\Ro^\mu)^*\omega_{Q/G}
+\pi_2^*\beta^\mu$. The critical curves $(y(t),\tilde\xi(t))$ of
the reduced system satisfy, with $\gamma(t)=(\dot
x(t),y(t),\tilde\xi(t))$
\[
\big(i_{\dot \gamma} \omega_{\Ro^\mu}  = -dE_{\Ro^\mu}
+f\big)|_{\gamma}.
\]\end{theorem}
In the case that $(Q,L,F)$ is conservative (i.e. $F=0$), we can
apply the constraint algorithm to show that there exists a
submanifold $S$ of the final constraint submanifold on which a
unique second-order vector field $\Gamma$ exists such that
\[
\left.\left(i_{\Gamma} \omega_{\Ro^\mu} =
-dE_{\Ro^\mu}\right)\right|_{S}.
\]
The integral curves of $\Gamma$ are lifted critical curves of the
Routhian reduced system.}
\end{remark}

\section{Reconstruction}\label{sec:reconstruction}

The reconstruction process deals with the problem of finding a
critical curve $q(t)$ of the Lagrangian system $(Q,L,F)$ such that
$\dot q(t)$ projects onto a {\em given} critical curve
$(y(t),\tilde \xi(t))$ of the intrinsically constrained Lagrangian
system $(M\to Q/G,\Ro^\mu,f+\zeta^\mu)$ with
$M=Q/G_\mu\times\tilde\lag\to Q/G$. Using Theorem~\ref{thm:routh},
it is sufficient to construct a curve $q(t)$ in $Q$ that projects
onto $(y(t),\tilde \xi(t))$ for all $t$, i.e. such that (i)
$[q(t)]_{G_\mu}=y(t)$, and (ii) $\tilde\xi(t) =
[q(t),\omega_{q(t)}(\dot q(t))]_G$.

To find such a curve $q(t)$ in $Q$, we choose an arbitrary
starting point $q_a \in Q$ such that $\pi_\mu(q_a)=y(a)$. Next,
consider the unique horizontal curve $q_h(t)$ through $q_a$ that
projects onto $x(t)$. If a curve $q(t)$ exists for which (i) and
(ii) holds, then it is a gauge of $q_h(t)$, that is,
$q(t)=q_h(t)g(t)$ for a certain curve $g(t)$ in $G$. The
reconstruction process then only consists of solving a first order
differential equation on $G$ to determine $g(t)$. For that
purpose, let $\xi (t)$ be the curve in $\mathfrak{g}$ such that
$[q_h(t),\xi(t)]_G=\tilde \xi(t)\in \tilde{\mathfrak{g}}$. Then,
condition (ii) is satisfied if and only if $g(t)$ solves the first
order differential equation
\[
\dot g(t)g^{-1}(t) = \xi(t),
\]
with $g(a)=e$. This determining condition for $g(t)$ comes from
the fact that the projection of $\dot q(t)$ onto $\tilde\lag$ has
to be $\tilde \xi(t)$.

It now remains to check that (i) is true for the curve
$q(t)=q_h(t)g(t)$, i.e. $y(t) = [q_h(t)g(t)]_{G_\mu}$. Fix a time
$t$, then the horizontal part of the tangent vector to the curve
$[q_h(t)g(t)]_{G_\mu}$ is the horizontal lift of $\dot x(t)$ to
the point $y(t)$. Similarly it follows that the vertical component
of the curve $[q_h(t)g(t)]_{G_\mu}$ is precisely $(y,\tilde \xi) +
\tilde \lag_\mu$. Since $g(a)=e$ the initial condition for
$[q(t)]_{G_\mu}$ is $y(a)$, we may conclude that $y(t)$ and
$[q(t)]_{G_\mu}$ coincide (both curves have the same horizontal
and vertical parts and the same initial conditions). In this sense,
the intrinsic constraint associated to the gyroscopic force
$\zeta^\mu$ guarantees that $(y(t),\tilde\xi(t))$ is the
projection of a lifted curve $\dot q(t)$ in $TQ$, i.e. it reflects
the second-order nature of the original Lagrangian system
$(Q,L,F)$. To conclude this section we unite
Theorems~\ref{thm:routh},~\ref{thm:routhsplit}
and~\ref{thm:presymp}.
\begin{theorem}\label{thm:routh2}
Any critical curve of the Lagrangian system $(Q,L,F)$ on $Q$ with
momentum $\mu$ projects in $M=Q/G_\mu\times\tilde\lag$ onto a
critical curve of the intrinsically constrained Lagrangian system
$(M\to Q/G,\Ro^\mu,f+\zeta^\mu)$. Conversely, any critical curve of $(M\to Q/G,\Ro^\mu,f+\zeta^\mu)$ is the projection of a critical curve of $(Q,L,F)$. The
Euler-Lagrange equations of motion for the intrinsically
constrained system are
\begin{eqnarray}
&&{\cal EL}(\Ro^\mu)^h\left(\ddot x,\dot y,\dot{\tilde\xi}\right)=
i_{\dot x}\tilde \Omega^\mu(y)-f(\dot x,\tilde\xi) \nonumber\\
&&j_l(\dot x,\tilde\xi)=\tilde \mu(y)\nonumber\\
&&\dot y^v = (y,\tilde \xi) + \tilde\lag_\mu. \nonumber
\end{eqnarray}
These equations are equivalently formulated in a presymplectic
setting. Let $\gamma(t)=(\dot x(t),y(t),\tilde\xi(t))$ in
$T_M(Q/G)$ be associated to a curve $(y(t),\tilde\xi(t))$ in $M$.
This curve is critical if and only if $\gamma(t)$ solves
\[
\big(i_{\dot \gamma} \omega_{\Ro^\mu}  = -dE_{\Ro^\mu}
+f\big)|_{\gamma}.
\]
\end{theorem}

\section{Routhian reduced systems: standard cases and examples}\label{sec:cases}

\subsection{The regular case}
The classical geometric description of Routh reduction (for
example \cite{Ar}, \cite{marsdenrouth}) mainly differs from the one described
above in that the unreduced Routhian is only defined in $J^{-1}_L
(\mu)$ thus providing a reduced Routhian independent of the
variable $\tilde{\xi}$. This is done for $T-V$ Lagrangians. This
classical setting is a special case of Theorem~\ref{thm:routh2}
when additional regularity conditions are assumed (the constraint
$\F \Ro^\mu=0$ is a regular configuration constraint). The idea
is, using the results from \S\ref{sec:regconf}, to eliminate the variable $\tilde\xi$ by means of the momentum equation
$j_l(v_x,\tilde\xi) =\tilde\mu(y)$. For that purpose we now assume
that the Lagrangian $L$ is both $G$-invariant and $G$-regular,
i.e., the mapping $\mathbb{F}_{\tilde{\xi}}l$ is invertible (see
\S \ref{suy} above). Let $\kappa_l: T(Q/G)\times\tilde\lag^*\to
T(Q/G)\times\tilde\lag$ be the inverse of
$\mathbb{F}_{\tilde{\xi}}l$.

There is a subtlety in the application of the result from
\S\ref{sec:regconf}. Until now we have applied the results from
\S\ref{sec:lagsyst} to Routh reduction by identifying $M$ with
$Q/G_\mu\times \tilde\lag$ and $N$ with $Q/G$. In the following we
deviate from this and we let $N$ be $Q/G_\mu$. In other words, we
reinterpret the constrained system $(M\to Q/G, \Ro^\mu,
f+\zeta^\mu)$ as a constrained system $(M\to
Q/G_\mu,\Ro^\mu,f+\zeta^\mu)$, i.e. for now we forget the
constraints in $y$ defined from the fact that $\Ro^\mu$ depends on
$y$ but not on $\dot{y}$. Then, from \S\ref{sec:regconf} it follows that
the intrinsic constraint $\F_{\tilde \xi}\Ro^\mu=0$ is a regular
configuration constraint, and using the inverse $\kappa_l$ we find
a section $\gamma : T(Q/G_\mu )\to Q/G_\mu \times \tilde{\mathfrak{g}}$ in the sense of
\S\ref{sec:regconf} given by
\begin{equation}\label{guti}
\gamma (v_y) = (y,\kappa^2_l (v_x,\tilde{\mu}(y))),
\end{equation}
where $v_x= T\overline\pi _\mu (v_y)$ with $\overline\pi _\mu : Q/G_\mu \to Q/G$
the natural projection; and $\kappa^2_l=\mathrm{pr}_2\circ\kappa_l$ with $\mathrm{pr}_2:T(Q/G)\times
\tilde\lag \to \tilde\lag$ the projection onto the second factor. As in \S\ref{sec:regconf}, we are now able
to consider the Lagrangian system on $N=Q/G_\mu$. However, recall that we temporarily
forgot the intrinsic constraints in the variable $y$. We can
reintroduce these constraints by making the remark that the
resulting system on $N$ carries intrinsically constraints on the
fibration $N=Q/G_\mu\to Q/G$. For notational convenience we
directly introduce the Lagrangian system on $N$ as being a
constrained system. This systems reads $(Q/G_\mu\to Q/G,\overline\Ro^\mu,
\overline f+ \zeta^\mu)$, where
\begin{eqnarray*}
&&\overline\Ro^\mu: T_{Q/G_\mu}(Q/G)\to \R:(v_x,y) \mapsto
\Ro^\mu\big(v_x,y,\kappa^2_l(v_x,\tilde\mu(y))\big)\\
&&\overline f:T_{Q/G_\mu}(Q/G)\to T^*(Q/G):(v_x,y)\mapsto
f\big(\kappa_l(v_x,\tilde\mu(y))\big)
\end{eqnarray*}
and $\zeta^\mu$ as usual the gyroscopic force associated to the
2-form $\beta^\mu$ on $Q/G_\mu$. We conclude that the critical
curves of the constrained Lagrangian system $(M\to
Q/G,\Ro^\mu,f+\zeta^\mu)$ are in a one-to-one correspondence to
the critical curves of $(Q/G_\mu\to Q/G,\overline\Ro^\mu,\overline
f+ \zeta^\mu)$.

Recall that $\beta^\mu$ can be decomposed into a horizontal and
vertical part: $(\tilde\Omega,-\mathrm{ad}^*)$. Similar to
previous definitions, we let $\partial_y\overline \Ro^\mu$ be the
shorthand notation for the restriction of the Euler-Lagrange
operator ${\cal EL}(\overline\Ro^\mu)$ to vectors vertical to
$Q/G_\mu\to Q/G$. From Theorem~\ref{thm:routh2} we easily have:

\begin{theorem}\label{ulth}
Any critical curve of the Lagrangian system $(Q,L,F)$ with
momentum $\mu$ projects in $Q/G_\mu$ onto a critical curve of the
intrinsically constrained Lagrangian system $(Q/G_\mu\to
Q/G,\overline\Ro^\mu,\overline f+\zeta^\mu)$, and vice versa. The
Euler-Lagrange equations of motion for the intrinsically
constrained system are
\begin{eqnarray}
&&{\cal EL}(\overline \Ro^\mu)^h\left(\ddot x,\dot y\right)=
i_{\dot x}\tilde \Omega^\mu(y)-\overline f(\dot x), \nonumber\\
&&\partial _y \overline{\mathcal{R}}^\mu(\dot x,y) =
-\mathrm{ad}^* _{\dot{y}^v}\tilde{\mu}(y) \nonumber
\end{eqnarray}
\end{theorem}
The equation $\partial_y\overline{\mathcal{R}}^\mu(\dot x,y) =
-\mathrm{ad}^* _{\dot{y}^v}\tilde{\mu}(y)$ can be rewritten as
\begin{equation}\label{eqn:regvel} \dot y^v=(y,\kappa_l^2(\dot
x,\tilde\mu(y))) + \tilde \lag_\mu.
\end{equation}

This theorem is similar to the classical geometric formulation of
Routh reduction as in e.g.~\cite{jalna,marsdenrouth} where the
variational problem is confined to the set $J^{-1}_L(\mu)$. It is
known that critical curves in $TQ$ with conserved momentum $\mu$
are in one-to-one correspondence to critical curves of the
variational problem under the constraint $J_L^{-1}(\mu)\subset
TQ$. In these references the {\em Routhian} is defined as the
restriction of $R^\mu$ to $J_L^{-1}(\mu)$, and since this
restriction $R^\mu |_{J_L^{-1}(\mu )}$ is invariant with respect
to the action of $G_\mu$, it is reducible to
$J^{-1}_L(\mu)/G_\mu$. Using the fact that in the $G$-regular case
(cf. Proposition~\ref{porron})
\[
J^{-1}_L(\mu )/G_\mu \simeq T(Q/G)\times Q/G_\mu,
\]
it is straightforward to show that this reduced {\em Routhian} is
precisely $\overline\Ro^\mu$.

In~\cite{marsdenrouth} (as well as in the original work of Routh)
the Lagrangian $L$ is of the type $L=T_2-V$ with $T_2=\frac12\rho
(v_q,v_q)$ a kinetic energy associated with a non-degenerate
positively definite invariant metric $\rho$ on $Q$ and $V:Q\to
\mathbb{R}$ an invariant potential energy.
Following~\cite{CMR01,jalna} one chooses the principal connection
to be the mechanical connection on $Q$ (the connection whose
horizontal distribution is orthogonal to $V\pi$ w.r.t the kinetic
energy metric $\rho$). The induced metric on $Q/G$ is denoted by
$\overline \rho$ (defined by $\overline
\rho(v_x,v_x)=\rho(v_x^h,v_x^h)$) and the vertical part of the
metric defines a metric $\tilde I$ on $\tilde\lag$. This is called
the inertia tensor and we assume it to be regular (i.e. the
Lagrangian is $G$-regular). The momentum map then equals
$j_l(x,\tilde\xi)=\flat_{ \tilde I}(x)(\tilde \xi) =\tilde \mu(y)$
and the Routhian on $Q/G_\mu$ is
\[
\overline \Ro^\mu (v_x,y) = \frac12\overline \rho(v_x,v_x)
-\left(V(x)+\frac12\langle\tilde\mu(y),\sharp_{\tilde
I}(x)(\tilde\mu(y))\rangle\right).
\]
One should note that this Routhian may differ from the one found
in classical books dealing with Routhian reduction. The difference
lies in the fact that one typically works in a local coordinate
chart, where the choice of the connection $\omega$ is taken to be
the local connection with vanishing connection coefficients (and
thus zero curvature). In the reduced system the curvature
force-term vanishes. This choice of connection with vanishing
curvature will imply that the Routhian contains linear terms in
the velocity, see also Remark~\ref{rem:altrouth}.

For the reconstruction process in this regular situation one may proceed
as in \S\ref{sec:reconstruction}. Indeed, given a solution $y(t)$ of
the reduced problem defined by $\overline \Ro^\mu $ we consider the curve
$\gamma (\dot{y}(t))$ in $Q/G_\mu \times
\tilde\lag$, where $\gamma $ is as in~(\ref{guti}). Then a family
of critical curves $q(t)$ of the problem defined by $L$ can be recovered
as in the general case. Alternatively (see \cite[\S IV]{marsdenrouth}), the curve $q(t)$
can be recovered by considering any curve $\bar{q}(t)$ with $y(t)=[\bar{q}(t)]_{G_\mu}$
and writing  $q(t)=\bar{q}(t)g(t)$, where $g(t)\in G_\mu$ is obtained by requiring that
$q(t)\in J_L^{-1}(\mu)$,$\forall t$.

\subsubsection{The (pre)-symplectic nature of classical Routhian reduction}
It was formulated in~\cite{marsdenrouth} how the symplectic bundle
picture for Routhian reduction looks like. Using
Proposition~\ref{prop:presymp} we can extend this to the more
general case of presymplectic geometry. The Routhian reduced
equations of motion in the classical case (so with a positive
definite kinetic energy) are intrinsically constrained
Euler-Lagrange equations living on $T_{Q/G_\mu}(Q/G)= T(Q/G)
\times Q/G_\mu$. Using the notations from
Section~\ref{sec:lagsyst} we have $M=Q/G_\mu$, $N=Q/G$. We can say
that critical curves of the Routh reduced system are solutions to
the following presymplectic system, with $\gamma(t)=(\dot
x(t),y(t))$:
\[
\left(i_{\dot \gamma (t)
}\big((\overline\pi_1\circ\F_1\overline\Ro^\mu)^*\omega_N +
\pi^*_2\beta^\mu\big)  = -dE_{\overline\Ro^\mu}+\overline
f\right)|_{\gamma},
\]
where $\mathbb{F}_1\overline{\mathcal{R}}^\mu \colon T(Q/G)\times
Q/G_\mu \to T^* (Q/G)\times Q/G_\mu$ stands for the fiber
derivative along the $T(Q/G)$ factor only. In the case that $\F _1
\overline{\mathcal{R}}^\mu$ is a diffeomorphism then the
presymplectic form is {\em symplectic}. This is the case if $L$ is
a regular Lagrangian, see also~\cite{BC}.

\subsubsection{The modified Tippe Top}

The purpose of this example is to demonstrate the relevance of
Routh reduction for nonconservative systems. In~\cite{CL} a Tippe
Top was described as a Lagrangian system with a dissipative force
term. The system is invariant under the action of $S^1$ and a
local stability analysis for relative equilibria was proven by
means of Routh reduction. We only mention here the setup of this
particular Lagrangian system within theoretical framework
developed in this article. The configuration space is $Q=SO(3)$
and we will work in local coordinate chart determined by the Euler
angles $(\varphi,\theta,\psi)$.  We consider the Abelian group
$G=S^1$ and the action on $SO(3)$ is defined by, for $\alpha \in
S^1$:
\[
(\varphi,\theta,\psi)\mapsto(\varphi+R\alpha,\theta,\psi-\e\alpha),
\]
for $\epsilon , R >0$. Next, we define the Lagrangian $L$ which is
of mechanical type and a dissipative force term $F$ which
represents the friction due to the slipping of the contact point
of the Tippe Top with the horizontal plane on which it evolves:
\begin{eqnarray*} L&=&\frac12 \left( (\e^2m\sin^2\theta+A)\dot
\theta^2 + A\sin^2\theta\dot\varphi^2+ C(\dot\psi
+\dot\varphi\cos\theta)^2\right) -
mg(R-\e\cos\theta),\\
F &=& -\mu (R-\e\cos\theta)^2\dot\theta d\theta
 -\mu \e \sin^2\theta(\e\dot \varphi+R\dot\psi)d\varphi
-\mu  R\sin^2\theta(R\dot \psi +\e\dot\varphi)d\psi,
\end{eqnarray*}

It is straightforward to see that $(Q,L,F)$ is invariant. The
associated momentum map: $TQ\to \R$ is:
\begin{eqnarray*}
J_L&=&\F L(q,\dot q)(R\partial_\phi-\e\partial_\psi) =
RA\sin^2\theta \dot \varphi +C(\dot\psi+\cos\theta\dot\varphi)
(R\cos\theta-\e)\\
&=& R\left((A\sin^2\theta+C\cos^2\theta)\dot \varphi-
C\e\dot\psi\right) +C\cos\theta(R\dot\psi-\e\dot\varphi),
\end{eqnarray*}
also called the Jellet integral. Fix a value $\mu\in \R$ of $J_L$.
The group is abelian, therefore $G_\mu =S^1$, and $SO(3)/S^1 \cong
S^2$. In order to define the Routhian, we consider the following
splitting of $TSO(3)$: the horizontal subspace is defined by the
(invariant) subspace spanned by
\[
\partial_\theta \mbox{ and } \e\partial_\varphi+R\partial_\psi.
\]
The curvature is zero, and we know that the Routh reduced system
is a Lagrangian system on $S^2$, i.e. of the form ${\cal
EL}(\overline \Ro^\mu) = -\overline f$, cf. Theorem~\ref{ulth}
with $\tilde\Omega=0$. In this case, the Routhian becomes:
\[
R^\mu=L-\mu\left(\frac{1}{\e^2+R^2}(R\dot\varphi-\e\dot\psi)\right)
\]
The next step is to use the Jellet integral and to compute
$\overline R^\mu$. Due to the complexity of the computations, we
refer the reader to~\cite{CL}.

\subsubsection{The free rigid body}\label{ex:freerigid}

We discuss the example of a free rigid body to illustrate the
meaning of the equation $\partial_y\overline\Ro^\mu = -
\mathrm{ad}^*_{\dot y^v}\tilde\mu(y)$ in Theorem~\ref{ulth}. We
have shown that it corresponds to $\dot
y^v=(y,\tilde\xi)+\tilde\lag_\mu$, with $\tilde\xi$ determined
from $j_l(\dot x, \tilde\xi)=\tilde\mu(y)$ and here we
illustrate this correspondence explicitly.

Following~\cite{marsdenrouth}, the configuration space of the free
rigid body is the entire group $Q=G=SO(3)$ on which $G$ acts from
the left (in view of remark~\ref{rem:leftright} we should change
the sign of the 2-form: $+\mathrm{ad^*}$). An element in $SO(3)$
is denoted by $A$ and corresponds to the rotation taking a
reference configuration (with principle inertia axis) of the rigid
body to its configuration at time $t$.  The Lagrangian is the left
invariant kinetic energy and the momentum equation corresponds to
the spatial angular momentum. Using the notations
from~\cite{marsdenrouth}, we assume that the the fixed momentum
equation corresponds to $\pmb\pi=\mu{\bf k}$ (where ${\pmb \pi}$
is the spatial angular momentum and $\bf k$ the unit vector in the
positive $z$-axis). We choose the standard connection
$\omega(A)(\dot A) = \dot AA^{-1}$, and we identify the
(trivial) bundle $\tilde\lag$ with $\lag$ by $[A,\xi]_G \mapsto
A^{-1}\xi A$, i.e. the representant of $[A,\xi]_G$ at unity. The
map determined by the connection $\omega$ which is used to identify $TSO(3)/G$ with
$T(Q/G)\times \tilde\lag \cong \lag$ becomes $(A,\dot
A)\mapsto A^{-1}\dot A$, i.e. the image is the angular velocity in
the body reference frame. In a local coordinate neighborhood
determined by Euler angles, the projection $TSO(3)\to \lag\cong\R^3$ equals
\[
(\phi,\theta,\psi,\dot \phi,\dot \theta,\dot \psi)\mapsto
(\xi_1,\xi_2,\xi_3)^T
\]
with  \begin{eqnarray}
  \xi_1 &=& \dot\theta\cos\psi+\dot\phi\sin\theta\sin\psi\nonumber \\
  \xi_2 &=& -\dot\theta\sin\psi+\dot\phi\sin\theta\cos\psi\nonumber \\
  \xi_3 &=& \dot\psi+\dot\phi\cos\theta.\label{eq:eulerangle}
\end{eqnarray} The reduced Lagrangian $l$ then determines
a function on $\lag$ and equals $l=\frac12 (I_1\xi^2_1+I_2\xi^2_2+I_3\xi^2_3)$,
with $(I_1,I_2,I_3)$ the inertia tensor in the body reference frame.

The isotropy subgroup $G_\mu = S^1$ consists of the rotations
about the $\bf k$-axis;  therefore $\pi_\mu: Q\to Q/G_\mu =S^2:
A\mapsto A^{-1} {\bf k}$. In terms of the Euler angles, this
projection is $(\phi,\theta,\psi)\mapsto (\theta,\psi)$. The map
$\tilde\mu: Q/G_\mu \to \tilde\lag^*\cong\lag^*$ is the
momentum $\mu{\bf k}$ expressed in body coordinates:
\[
(\theta,\psi)\mapsto (\mu \sin\theta\sin\psi,\mu
\sin\theta\cos\psi,\mu\cos\theta).
\]
The configuration space of the Routh reduced system is simply
$Q/G_\mu = S^2$ and the only meaningful equation from
Theorem~\ref{ulth} is $\partial_y\overline\Ro^\mu =
+\mathrm{ad}^*_{\dot y^v}\tilde\mu(y)$.

In the following we will first compute the equation $\xi =
\kappa^2_l(\tilde\mu(y))$ (see~(\ref{guti})) and study its projection to the tangent
bundle of $Q/G_\mu = S^2$. Secondly, we show that this projection
corresponds to the equation $\partial_y\overline\Ro^\mu =
+\mathrm{ad}^*_{\dot y^v}\tilde\mu(y)$.

To compute $\kappa_l^2$ we start with the reduced momentum
equation $j_l(\tilde \xi)=\tilde \mu(y)$:
\begin{eqnarray}
I_1\xi_1&=& \mu \sin\theta\sin\psi,\nonumber\\
I_2\xi_2&=& \mu \sin\theta\cos\psi,\nonumber\\
I_3\xi_3&=&\mu\cos\theta.\label{ex:free}
\end{eqnarray}
It should be clear that since $I_i\neq 0$, we can compute $\xi$
explicitly
\begin{equation}\label{eq:invxi}
\xi_1= \frac{1}{I_1}\mu \sin\theta\sin\psi,\ \xi_2=
\frac{1}{I_2}\mu \sin\theta\cos\psi,\
\xi_3=\frac{1}{I_3}\mu\cos\theta.
\end{equation}

Next, we need to compute the corresponding element in $TS^2$. In
the general theory this is done by identifying $TS^2$ with the
quotient $\overline\pi^*_\mu\tilde\lag/\tilde\lag_\mu$. In terms
of the Euler-angles, this boils down to computing $\dot \theta$
and $\dot \psi$ out of~(\ref{ex:free}): $\dot
\theta=\xi_1\cos\psi-\xi_2\sin\psi$ and $\dot \psi=\xi_3\sin\theta
-(\xi_1\sin\psi+\xi_2\cos\psi)\cos\theta$. If we substitute the
values for $\xi_{1,2,3}$ in these two expressions, we find the
equation $\dot y^v=(y,\tilde\xi)+\tilde\lag_\mu$, with $\tilde\xi$
determined from $j_l(\tilde\xi)=\tilde\mu(y)$:
\begin{eqnarray*}
\dot \theta &=&\mu\sin\theta\sin\psi\cos\psi\left(\frac{1}{I_1}-\frac{1}{I_2}\right),\\
\dot\psi &=&\mu\cos\theta\left(\frac{1}{I_3}
-\left(\frac{\sin^2\psi}{I_1}+\frac{\cos^2\psi}{I_2}\right)\right),
\end{eqnarray*}

The second way to retrieve these equations is using the Routhian
$\overline \Ro^\mu$. By definition it is obtained from $\Ro^\mu =
l (\xi)  - \langle \tilde \mu(y)  ,\xi\rangle$, where $\xi$ is
substituted by~(\ref{eq:invxi}):
\[
\overline\Ro^\mu(\theta,\psi) =
-\frac{\mu^2}{2}\left(\frac{\sin^2\theta\sin^2\psi}{I_1}
+\frac{\sin^2\theta\cos^2\psi}{I_2}+\frac{\cos^2\theta}{I_3}\right).
\]
Next, we compute the 2-form $\mathrm{ad}^*=\beta^\mu$ on $TS^2$.
Recall that it is the projection of $d\omega^\mu$ to $Q/G_\mu$. In
Euler angles, we have that $\langle\mu,\omega\rangle (A,\dot A) =
\tilde\mu(y)\cdot (\xi)^T = \mu(\dot \phi +\cos\theta \dot\psi)$.
Therefore $\beta^\mu = -\mu \sin\theta d\theta\wedge d\psi$, the
gyroscopic force term is $\zeta^\mu = -i_{\dot y} \beta^\mu =-
\mu\sin\theta \dot \psi d\theta + \mu\sin\theta \dot \theta d\psi$
and the Euler-Lagrange equations for the Routh reduced system are
${\cal EL}(\overline \Ro^\mu) = -\zeta^\mu$:
\begin{eqnarray*}
\fpd{\overline\Ro^\mu}{\psi}= -
\mu^2\sin^2\theta\sin\psi\cos\psi\left(\frac{1}{I_1}-\frac{1}{I_2}\right)&=& -\mu\sin\theta\dot \theta
 \\
\fpd{\overline\Ro^\mu}{\theta} =
-\mu^2\sin\theta\cos\theta\left(\left(\frac{\sin^2\psi}{I_1}+
\frac{\cos^2\psi}{I_2}\right)-\frac{1}{I_3}\right)&=&\mu\sin\theta\dot
\psi.
\end{eqnarray*}
It should be clear that these Euler-Lagrange equations, when
brought into normal form, are precisely the equations of motion
obtained above by projecting $\tilde\xi = \kappa^2_l(\tilde\mu(y))$
to $TS^2$.

\subsection{Lagrangian systems subjected to magnetic forces}\label{ssec:gyro}
Assume that we study a conservative Lagrangian system $L$, where
$L=T_2+T_1-V$, $T_2=\frac12\rho(v_q,v_q)$ is the kinetic energy
and $T_1 = \langle\alpha(q),v_q\rangle$ with $\alpha$ and
invariant 1-form. These linear terms in the kinetic energy are
found, for example, when the mechanical system is subject to
magnetic force terms (for instance, charged particle/rigid body
evolving in a magnetic field). The restriction of $\alpha$ to
horizontal vectors is denoted by $\alpha^h$ and is projectable to
a one-form $\overline\alpha$ on $Q/G$. The restriction of $\alpha$
to verticals determines a section $\tilde\alpha^v$ of
$\tilde\lag^*\to Q/G$. We consider the mechanical connection
$\omega$ defined by the metric $\rho$ as before, and let $\tilde
I$ be a regular inertia-tensor (i.e. $L$ is $G$-regular). The
momentum condition then reads
\[
\flat_{\tilde I}\tilde\xi + \tilde \alpha^v = \tilde\mu,
\]
with inverse $\tilde\xi=\sharp_{\tilde
I}(\tilde\mu-\tilde\alpha^v)$. The Routhian on $Q/G_\mu$ assumes
the form
\[
\overline \Ro^\mu (v_x,y) = \frac12\overline
\rho(v_x,v_x)+\langle\overline\alpha(x),v_x\rangle
-\left(V(x)+\frac12\langle(\tilde\mu-\tilde\alpha^v),
\sharp_{\tilde I}(\tilde\mu-\tilde\alpha^v)\rangle\right).
\]

\subsubsection{The heavy top in a magnetic field}

We consider a constant magnetic field along the $z$-axis in the
inertia frame. The Lagrangian of a charged rigid body with a fixed
point is invariant under the left action of rotations about the
$z$-axis (i.e. when $(\phi,\theta,\psi)$ is a coordinate chart
determined by the Euler-angles, then $\phi$ is (locally) cyclic).
Thus $Q=SO(3)$, $G_\mu=G=S^1$ and $Q/G = S^2$, see
also~\cite{marsdenrouth}. The Lagrangian equals, with
$\xi_i$ as written down in Eq.~(\ref{eq:eulerangle}):
\[
L=\frac12 \sum_{i=1}^3 I_i\xi_i -m g A^{-1} {\bf k}\cdot {\pmb
\chi} - \Omega B A^{-1}{\bf k}\cdot I{\pmb \xi},
\]
with $\Omega = q/mc$, with $q$ the charge of the heavy top, $g$
the earth's acceleration, $m$ the mass and ${\pmb \chi}$ the
vector with length $\epsilon$ connecting the fixed point with the
center of mass in the body reference frame. We assume that the top
is symmetric and that the moving frame is directed along its
symmetry axis (i.e. $I_1=I_2$). In this case, the Lagrangian
becomes
\[
L = \frac12 I_1\dot\theta^2 + \frac12I_3\dot \psi^2 +
I_3\cos\theta\dot\phi\dot\psi + \frac12\rho_{\phi\phi} \dot \phi^2
- mg\epsilon\cos\theta -\Omega B (\rho_{\phi\phi}\dot \phi +
I_3\cos\theta\dot\psi),
\]

with $\rho_{\phi\phi}=I_1\sin^2\theta +I_3\cos^2\theta$. Thus
$\alpha = -\Omega B (\rho_{\phi\phi}d\phi+ I_3\cos\theta d\psi)$, $\tilde \alpha^v=-\Omega B \rho_{\phi\phi}$ and $\alpha^h=0$.
Here we have chosen $\omega$ to be the mechanical connection
$\omega=d\phi+(I_3\cos\theta d\psi)/\rho_{\phi\phi}$. The
conserved momentum $J_L$ reads
\[
\rho_{\phi\phi} \dot \phi + I_3\cos\theta\dot\psi -\Omega B
\rho_{\phi\phi} = \mu.
\]

The next step is the computation of the Routhian
$\overline\Ro^\mu$, which can be carried out in two ways: either
one computes $\dot \phi$ out of the momentum equation, and then
substitute the expression in $L-\mu(\dot\phi
+(I_3\cos\theta\dot\psi)/\rho_{\phi\phi})$, or one computes
$\overline\rho$ and use the expression of $\overline\Ro^\mu$
determined above. After straightforward but tedious
computations one finds that the Routhian $\overline\Ro^\mu$ for
the magnetic heavy top assumes the following form
\[
\overline\Ro^\mu = \frac12\left(I_1 \dot \theta^2 +
\frac{I_3I_1\sin^2\theta}{\rho_{\phi\phi}}\dot\psi^2\right) -
V_\mu(\theta)
\]
with
\[
V_\mu(\theta)=mg\cos\theta+\frac12\frac{(\mu+\Omega
B\rho_{\phi\phi})^2}{\rho_{\phi\phi}}.
\]
In the case that $\Omega B \ll 1$, we may neglect the term in
$(\Omega B)^2$ and then the effective potential $V_\mu$ equals up
to constant the effective potential for the heavy top (i.e. set
$B=0$). In this case the critical curves of both problems
coincide, but the reconstruction process will differ. The
reconstruction equation, roughly said, is the inverse of the
momentum map:
\[
\dot \phi = \frac{\mu-I_3\cos\theta\dot\psi}{\rho_{\phi\phi}} +
\Omega B
\]
This result is also known as Larmor's theorem (for example, see
\cite{Gold}): the motion of a charged particle in a constant
magnetic field is the motion of the particle in the absence of the
magnetic field superposed with a rotation with constant frequency
$\Omega B$, the Larmor frequency.

\subsection{Lagrangians linear in the symmetry generators}
Here we study the case where the Lagrangian of an invariant
conservative Lagrangian system $(Q,L,0)$ generates a reduced
momentum map $j_l$ that does not depend on the variable
$\tilde\xi$, i.e. for a critical curve $(y(t),\tilde\xi(t))$
projecting onto $x(t)$:
\[
j_l(\dot x)=\tilde \mu(y).
\]
This situation occurs for instance if one studies Lagrangian
systems for which the kinetic energy is not strictly positive
definite, i.e. if $L(v_q)=\frac12 \rho(q)(v_q,v_q) - V(q)$, with
$\rho$ and $V$ right invariant, $\rho$ a symmetric two-tensor such
that $\rho(\sigma(\xi),\sigma(\xi)) \equiv 0$ for any
$\xi\in\lag$.

In this situation we can consider the metric $\overline\rho$ on
$Q/G$ defined in the usual way as $\overline\rho(v_x,v_x)=
\rho(q)((v_x)_q^h,(v_x)_q^h)$, for $\pi(q)=x$ and $v_x$ arbitrary
in $T_x(Q/G)$. The remaining part of $\rho$ yields a
$\tilde\lag^*$-valued one-form $\tilde\rho$ on the quotient space
$Q/G$.  The momentum constraint for a critical curve then yields
$\tilde\rho(x)(\dot x) =\tilde\mu(y)$.

We can apply Theorem~\ref{thm:routhsplit} and we find that the
Routhian $\Ro^\mu(v_x,y,\tilde\xi)$ equals
\[
\frac12\overline\rho(x)(v_x,v_x) - V(x) +
\langle(\tilde\rho(x)-\tilde\mu(y)),\tilde\xi\rangle.
\]
From Section~\ref{sec:lagsyst}, Proposition~\ref{prop:lagcon} we
know that the critical curves of this intrinsically constrained
Lagrangian system $(M\to Q/G, \Ro^\mu,\zeta^\mu)$ project onto
critical curves of the Lagrangian system
$(Q/G_\mu,\overline\Ro^\mu,\zeta^\mu)$ constrained to the
submanifold in $T(Q/G_\mu)$ determined by $\tilde\rho(x)(v_x)
=\tilde\mu(y)$. The function $\overline\Ro^\mu$ is defined as
\[
\overline\Ro^\mu(v_x)= \frac12\overline\rho(x)(v_x,v_x) - V(x) ,
\]
where we identified $V$ with its reduced function on $Q/G$. This
particular situation is encountered in general relativity, where
solutions to Einstein's equations admit lightlike Killing vectors.

\subsubsection{Lightlike Killing vectors}
This example shows how the presented theory can be applied to
reduce geodesic equations in general relativity where the metric
admits a lightlike Killing vector. We assume that the Killing
vector $X$ is complete, i.e. the flow of $X$ determines an Abelian
group action of $\R$ on the spacetime $(M,g)$. Since $X$ is
Killing it determines a symmetry of the kinetic energy Lagrangian
$L(v_q)=\frac12 g(v_q,v_q)$.

A relevant local example from general relativity illustrating the
above can be found in the so called pp-wave solution to Einstein's
equations (see~\cite{kramer}):
\[
g = H(u,x,y)du^2 +2dudv+dx^2+dy^2,
\]
expressed in the Brinckmann coordinates. The vector field
$\partial_v$ is Killing, and the momentum constraint then assumes
the form $J_L(u,x,y,\dot u) = \dot u=\mu$. The geodesics with
momentum $\mu$ are critical curves for the variational system on
$T\R^3$ with Lagrangian equal to
\[
\overline\Ro^\mu=\frac12\left(\mu^2 H(u,x,y)+\dot x^2+\dot
y^2\right)
\]
constrained to the submanifold $\dot u=\mu$.

\section*{Acknowledgements} The authors want to thank Tudor
S. Ratiu for useful conversations about some
examples where the classical approach can not be carried out. We
are also indebted to Frans Cantrijn for the careful reading of
this paper. MCL was partially supported by Ministerio de Ciencia y
Tecnolog\'{\i}a of Spain under grants MTM2007-60017 and
MTM2008-01386 and BL by a Research Programme of the Research
Foundation - Flanders (FWO).

\end{document}